\ifx\synctex\undefined\else\synctex=1\fi
\documentclass[a4paper,UKenglish]{lipics}

\usepackage{pdfsync}
\usepackage{microtype}
\usepackage{amsmath,amssymb,latexsym,amsfonts}
\usepackage[only,llbracket,rrbracket]{stmaryrd}
\usepackage{extarrows}
\usepackage{ifthen}
\usepackage[textsize=small]{todonotes} 
\usepackage{paralist}
\usepackage{xifthen}
\usepackage{subfig}
\usepackage{colonequals}
\usepackage{mathabx}
\usepackage{xspace}
\usepackage{thmtools,thm-restate}
\usepackage{framed}
\usepackage{tikz-qtree}

\usepackage{hyperref}

\hypersetup{
    bookmarks=true,         
    unicode=false,          
    pdftoolbar=true,        
    pdfmenubar=true,        
    pdffitwindow=false,     
    pdfstartview={FitH},    
    pdftitle={My title},    
    pdfauthor={Author},     
    pdfsubject={Subject},   
    pdfcreator={Creator},   
    pdfproducer={Producer}, 
    pdfkeywords={keyword1} {key2} {key3}, 
    pdfnewwindow=true,      
    colorlinks=true,        
    linkcolor=black,        
    citecolor=black,        
    filecolor=magenta,      
    urlcolor=cyan           
}

\usepackage{thmtools,thm-restate}

\usepackage{tikz}
\usetikzlibrary{arrows,shapes,snakes,automata,backgrounds,fit,positioning}

\tikzset{
  treenode/.style = {align=center, inner sep=0pt, text centered,
    font=\sffamily},
  arn_n/.style = {treenode, circle, white, font=\sffamily\bfseries, draw=black,
    fill=black, text width=1.5em},
  arn_r/.style = {treenode, circle, red, draw=red, 
    text width=1.5em, very thick},
  arn_x/.style = {treenode, rectangle, draw=black,
    minimum width=0.5em, minimum height=0.5em}
}

\input{defs}

\usepackage{environ}
\NewEnviron{killcontents}{}

\title{Ordered Tree-Pushdown Systems}

\author[1]{Lorenzo Clemente\thanks{This work was partially supported by the National Science Center (decision DEC-2013/09/B/ST6/01575)}}
\author[2]{Paweł Parys\thanks{This work was partially supported by the National Science Center (decision DEC-2012/07/D/ST6/02443)}}
\author[3]{Sylvain Salvati}
\author[4]{Igor Walukiewicz\thanks{This work was partially supported by the Technische Universität München – Institute for Advanced Study, funded by the German Excellence Initiative and the European Union Seventh Framework Programme, grant  n. 291763.}}
\affil[1,2]{University of Warsaw, Poland}
\affil[3,4]{CNRS, Universit\'e de Bordeaux, INRIA, France}
\authorrunning{L. Clemente, P. Parys, S. Salvati, I. Walukiewicz} 

\Copyright{Lorenzo Clemente, Paweł Parys, Sylvain Salvati, Igor Walukiewicz}



\serieslogo{}
\volumeinfo
  {Billy Editor and Bill Editors}
  {2}
  {Conference title on which this volume is based on}
  {1}
  {1}
  {1}
\EventShortName{}
\DOI{10.4230/LIPIcs.xxx.yyy.p}

\begin{document}

\maketitle

\begin{abstract}
	We define a new class of pushdown systems where the pushdown is a tree instead of a word.
	We allow a limited form of lookahead on the pushdown conforming to a certain ordering restriction,
	and we show that the resulting class enjoys a decidable reachability problem.
	This follows from a preservation of recognizability result for the backward reachability relation of such systems.
	As an application, we show that our simple model can encode several formalisms generalizing pushdown systems,
	such as ordered multi-pushdown systems,
	annotated higher-order pushdown systems,
	the Krivine machine,
	and ordered annotated multi-pushdown systems.
	In each case, our procedure yields tight complexity.
\end{abstract}

\section{Introduction}

\paragraph{Context.}

Modeling complex systems requires to strike the right balance between the accuracy of the model, and the complexity of its analysis.
A successful example is given by \emph{pushdown systems},
which are a popular class of infinite-state systems arising in diverse contexts,
such as language processing, data-flow analysis, security, computational biology, and program verification.
Many interesting analyses reduce to checking reachability in pushdown systems,
which can be decided in \ptime using, e.g., the popular \emph{saturation technique}
\cite{BouajjaniEsparzaMaler:Pushdown:1997,FinkelWillemsWolper:Pushdown:1997}
(cf.{\@} also the recent survey \cite{CarayolHague:SaturationSurvey:2014}).
Pushdown systems have been generalized in several directions.
One of them are \emph{tree-pushdown systems}~\cite{Guessarian:TPDA:1983},
where the pushdown is a tree instead of a word.
Unlike for ordinary pushdown systems,
non-destructive lookahead on the tree pushdown leads to undecidability.
In this work we propose an ordering condition permitting a limited
non-destructive lookahead on a tree pushdown.

A seemingly unrelated generalization is \emph{ordered multi-pushdown systems} \cite{BreveglieriCherubiniCitriniCrespi-Reghizzi:Ordered:1996,AtigBolligHabermehl:Ordered:2008,Atig:ordered:2012},
where several linear pushdowns are available instead of just one.
Since already two unrestricted linear pushdowns can simulate a Turing machine,
an ordering restriction is put on popping transitions,
requiring that all pushdowns smaller than the popped one are empty.
Reachability in this model is \twoexptimec \cite{AtigBolligHabermehl:Ordered:2008}.

\emph{Higher-order pushdown systems} provide another type of
generalization. Here pushdowns can be nested inside other pushdowns
\cite{Maslov:Multilevel:1976,KnapikNiwinskiUrzyczyn:Easy:2002}.
\emph{Collapsible pushdown systems} \cite{KnapikNikinskiUrzyczynWalukiewicz:Panic:2005,HagueMurawskiOngSerre:Collapsible:2008}
additionally enrich pushdown symbols with \emph{collapse links} to inner sub-pushdowns.
This allows the automaton to push a new symbol and to save, at the same time, the current context in which the symbol is pushed,
and to later return to this context via a collapse operation.
\emph{Annotated pushdown systems} \cite{BroadbentCarayolHagueSerre:Saturation:2012} (cf.{\@} also \cite{Kartzow:Parys:MFCS:2012})
provide a simplification of collapsible pushdown systems
by replacing collapse links with arbitrary pushdown annotations%
\footnote{%
%
	Collapsible and annotated systems generate the same configuration graphs
	when started from the same initial configuration,
	since new annotations can only be created to sub-pushdowns of the current pushdown.
	However, annotated pushdown systems have a richer backward reachability set
	which includes non-constructible pushdowns.
}.
%
The \emph{Krivine machine} \cite{SalvatiWalukiewicz:Krivine:2011} is a related model
which evaluates terms in simply-typed $\lambda Y$-calculus.
Reachability in all these models is \nexptimec {(n-1)} \cite{BroadbentCarayolHagueSerre:Saturation:2012,SalvatiWalukiewicz:Krivine:2011}
(where $n$ is the order of nesting pushdowns/functional parameters),
and one exponential higher in the presence of alternation.
Even more general, \emph{ordered annotated multi-pushdown systems} \cite{Hague:FSTTCS:2013}
have several annotated pushdown systems under an ordering restriction similar to \cite{AtigBolligHabermehl:Ordered:2008} in the first-order case.
They subsume both ordered multi-pushdown systems and annotated pushdown systems.
The 
saturation method (cf.~\cite{CarayolHague:SaturationSurvey:2014}) has
been adapted to most of these models,
and it is the basis of the prominent MOPED tool~\cite{Esparza:Schwoon:MOPED:2001} for the analysis of
pushdown systems, as well as the C-SHORe model-checker for annotated pushdown systems \cite{BroadbentCarayolHagueSerre:C-SHORe:2013}.

\paragraph{Contributions.}

Motivated by a unification of the results above,
we introduce \emph{ordered tree-pushdown systems}.
These are tree-pushdown systems with a limited destructive lookahead on the pushdown.
We introduce an order between pushdown symbols,
and we require that, whenever a sub-pushdown is read, all sub-pushdowns of smaller order must be discarded.
%
%
The obtained model is expressive enough to simulate all the systems mentioned above,
and is still not Turing-powerful thanks to the ordering
condition.
%
Our contributions are:
\begin{inparaenum}[i)]
	\item A general preservation of recognizability result for ordered tree-pushdown systems.
	\item A conceptually simple saturation algorithm working on finite tree automata representing sets of configurations
	(instead of more ad-hoc automata models),
		subsuming and unifying previous constructions.
	\item A short and simple correctness proof.
	\item Direct encodings of several popular extensions of pushdown systems,
	such as ordered multi-pushdown systems, annotated pushdown systems, the Krivine machine,
	and ordered annotated multi-pushdown systems.
      \item Encoding of our model into Krivine machines with states,
        that in turn are equivalent to collapsible pushdown automata.
	\item A complete complexity characterization of reachability in ordered tree-pushdown systems and natural subclasses thereof.
\end{inparaenum}
%

\paragraph{Related work.} Our work can be seen as a generalization of
the saturation method for collapsible pushdown
automata~\cite{BroadbentCarayolHagueSerre:Saturation:2012} to a
broader class of rewriting systems. This method has been already
generalized in~\cite{Hague:FSTTCS:2013} to  multi-stack higher-order
systems; in particular for ordered, phase-bounded, and scope-bounded
restrictions. Another related work is a saturation method for recursive program
schemes~\cite{BroadbentK13}.  Schemes are equivalent to  $\lambda
Y$-calculus, so our formalism can be used to obtain a saturation
method for schemes. 

Ordered tree-pushdown systems proposed in the present paper unify these
approaches.
The encodings of the above mentioned systems are direct and work step-to-step.
By contrast, the encoding of the Krivine machine to higher-order
pushdowns is rather
sophisticated~\cite{HagueMurawskiOngSerre:Collapsible:2008,Salvati:Walukiewicz:MSCS:2015}, and even
more so its proof of correctness. The converse encoding of
annotated higher-order pushdowns into Krivine machines is conceptually
easier, but technically quite long for at least two reasons:
a state has to be encoded by a tuple of terms, and transitions of the
automaton need to be implemented with beta-reduction.

Concerning multi-pushdown systems, there exist restrictions that we do
not cover in this paper.  In~\cite{Hague:FSTTCS:2013} decidability is
proved for annotated multi-pushdowns with phase-bounded and
scope-bounded restrictions. For standard multi-pushdown systems,
split-width has been proposed as a unifying
restriction~\cite{CyriacGK12}.

\paragraph{Outline.}

In Sec.~\ref{sec:preliminaries} we introduce common notions. 
In Sec.~\ref{sec:otpds} we define our model
and we present our saturation-based algorithm to decide reachability.
In Sec.~\ref{sec:examples} we show that ordered systems can optimally encode several popular formalisms.
In Sec.~\ref{sec:safety} we discuss the notion of safety from the Krivine machine and higher-order pushdown automata,
and how it relates to our model.
In Sec.~\ref{sec:conclusions} we conclude with some perspectives on open problems.





\section{Preliminaries}

\label{sec:preliminaries}

We work with rewriting systems on ranked trees, and with alternating tree automata.
The novelty is that every letter of the ranked alphabet will have an order.
A tree has the order determined by the letter in the root.
The order itself is used to constrain rewriting rules.


An \emph{alternating transition system} is a tuple $\Sys = \tuple{\mathcal C, \goesto}$,
where $\mathcal C$ is the set of configurations
and $\goesto \subseteq \mathcal C \times 2^{\mathcal C}$ is the alternating transition relation.
%
For two sets of configurations  $A, B \subseteq \mathcal C$
we define $A \goesto_1 B$ iff, for every $c \in A$, either $c \in B$,
or there exists $C \subseteq B$ s.t.{\@} $c \goesto C$,
and we denote by $\goesto^*_1$ its reflexive and transitive closure. 
The set of \emph{predecessors} of a set of configurations $C \subseteq \mathcal C$ is $\prestar{C} = \setof c {\set c \goesto^*_1 C}$.

\paragraph{Ranked trees.}
Let $\Nat$ be the set of non-negative integers,
and let $\Natpos$ be the set of strictly positive integers.
A \emph{node} is an element $u \in \Natpos^*$.
A node $u$ is a \emph{child} of a node $v$ if $u = v\cdot i$ for some $i\in\Natpos$.
A \emph{tree domain} is a non-empty prefix-closed set of nodes $D \subseteq \Natpos^*$
s.t., if $u\cdot(i+1) \in D$, then $u\cdot i \in D$ for every $i \in \Natpos$.
A \emph{leaf} is a node $u$ in $D$ without children. 
%
A \emph{ranked alphabet} is a pair $(\Sigma, \rank)$
of a set of symbols $\Sigma$ together with a ranking function $\rank : \Sigma \to \Nat$.
%
A \emph{$\Sigma$-tree} is a function $t : D \to \Sigma$,
where $D$ is a tree domain,
s.t., for every node $u$ in $D$ labelled with a symbol $t(u)$ of rank $k$,
$u$ has precisely $k$ children.
%
%
For a $\Sigma$-tree $t : D \to \Sigma$ and a label $a \in \Sigma$,
let $t^{-1}(a) = \setof {u \in D} {t(u) = a}$
be the set of nodes labelled with $a$. For a tree $t$ and a node
$u$ therein, the \emph{subtree} of $t$ at $u$ is defined as expected. 
We denote by $\Trees \Sigma$ the set of $\Sigma$-trees.

%

\paragraph{Order of a tree.} In this paper we will give a restriction
on a tree rewriting system guaranteeing that $\prestar{C}$ is regular
for every regular set $C$. This restriction will use the notion of an
order of a tree. The order of a tree is simply determined by the
order of the symbol in the root. Therefore, we suppose that our alphabet
$\Sigma$ comes with a function $\ord:\Sigma\to\Nat$.
The \emph{order} of a tree $t$ is $\ord(t) := \ord(t(\varepsilon))$.

\paragraph{Rewriting.}

Let $\Vars_0,\Vars_1,\dots$ be pairwise disjoint infinite sets of
variables; and let $\Vars = \bigcup_n \Vars_n$.
We consider the extended alphabet $\Sigma \cup \Vars$
where a variable $\var x \in \Vars_n$ has rank $0$ and order $n$.
We will work with the set $\Trees {\Sigma, \Vars}$ of $(\Sigma \cup
\Vars)$-trees. For such a tree $t$, let $\vars t$ be the set of
variables appearing in it. We say that $t$ is \emph{linear} if each variable in $\vars t$ appears exactly once in $t$.
For some $(\Sigma \cup \Vars)$-tree $u$,
$t$ is \emph{$u$-ground} if $\vars t \cap \vars u = \emptyset$. 
%
%
A \emph{substitution} is a finite partial mapping
$\sigma : \Vars \to \Trees {\Sigma \cup \Vars}$ respecting orders, i.e.,~$\ord(\sigma(\var x)) = \ord(\var x)$.
Given a $(\Sigma \cup \Vars)$-tree $t$ and a substitution $\sigma$,
$t\sigma$ is the $(\Sigma \cup \Vars)$-tree obtained by replacing each variable $\varord x {} {}$ in $t$ in the domain of $\sigma$ with $\sigma(\varord x {} {})$.
%
A \emph{rewrite rule} over $\Sigma$ is a pair $l \rewritesto r$ of $(\Sigma \cup \Vars)$-trees $l$ and $r$
s.t.{\@} $\vars r \subseteq \vars l$ and $l$ is linear.\footnote
{Notice that we require that all the variables appearing on the r.h.s.{\@} $r$
also appear on the l.h.s.{\@} $l$.
All our results carry over even by allowing some variables on the r.h.s.{\@} $r$ not to appear on the l.h.s.{\@} $l$,
but we forbid this for simplicity of presentation.}
%
%


\ignore{
A \emph{root rewrite system} is a tuple $\Sys = \tuple {\Sigma, \Rules}$
where $\Sigma$ is an ordered alphabet
and $\Rules$ is a set of rewrite rules $l \rewritesto r$ over $\Sigma$.
A root rewrite system induces a rewrite relation $\rewritesto_\Sys$ on ordered trees by \emph{root rewriting}:
For two ordered trees $t_0$ and $t_1$,
we write $t_0 \rewritesto_\Sys t_1$ iff there exists a rewrite rule $l \rewritesto r \in \Rules$
and a substitution $\sigma$ s.t.{\@} $l = t_0\sigma$ and $r = t_1\sigma$.
Note that root rewriting is the top-down generalization to trees of prefix-rewriting over words.
}


\ignore{
A \emph{alternating root rewrite system with states} is a tuple $\Sys = \tuple {\Sigma, Q, p_0, F, \Rules}$
where $\Sigma$ is an ordered alphabet,
$Q$ is a finite set of control locations,
$\Rules$ is a set of rules of the form $p, l \rewritesto P, r$,
where $p \in Q$, $P \subseteq Q$, and $l \rewritesto r$ is a rewrite rule over $\Sigma$.
A \emph{configuration} is a pair $(p, t)$ where $p \in Q$ and $t$ is an ordered $\Sigma$-tree.
Let $\mathcal C$ be the set of configurations
$c_0$ is the initial configuration and $T$ a set of terminal configurations.
A root rewrite system induces an alternating transition system $\tuple {...}$
A root rewrite system with states induces a rewrite relation $\rewritesto_\Sys$ on sets of configurations:
For a configuration $(p, t)$ and a set of configurations $P \times \set {u}$,
we write $(p, t) \rewritesto_\Sys P \times \set {u}$
if, and only if, 
there exists a rule $(p, l) \rewritesto (P, r) \in \Rules$
and a substitution $\sigma$ s.t.{\@} $t = l\sigma$ and $u = r\sigma$.
%

The \emph{order} of a root rewrite system with states is the maximal order of rules in $\Rules$.
}

\paragraph{Alternating tree automata.}

An \emph{alternating tree automaton} (or just \emph{tree automaton})
is a tuple $\ATA A = \tuple {\Sigma, Q, \Delta}$
where 
$\Sigma$ is a finite ranked alphabet, 
$Q$ is a finite set of states,
and $\Delta \subseteq Q \times \Sigma \times (\pow Q)^*$
is a set of alternating transitions of the form $p \trans a P_1 \cdots P_n$,
with $a$ of rank $n$. 
We say that $\ATA A$ is \emph{non-deterministic} if, for every transition as above, all $P_j$'s are singletons,
and we omit the braces in this case.
An automaton is \emph{ordered} if, for every state $p$ and symbols $a, b$ s.t.{\@} $p \trans a \cdots$ and $p \trans b \cdots$,
we have $\ord (a) = \ord (b)$.
We assume w.l.o.g.{\@} that automata are ordered,
and we denote by $\ord (p)$ the order of state $p$.
The transition relation is extended to a set of states $P \subseteq Q$ by defining
$P \trans a P_1 \cdots P_n$ iff, for every $p \in P$, there exists a transition
$p \trans a P_1^p \cdots P_n^p$,
and $P_j = \bigcup_{p \in P} P_j^p$ for every $j\in\{1,\dots,n\}$.
%
%
%
It will be useful later in the definition of the saturation procedure
to define run trees not just on ground trees,
but also on trees possibly containing variables.
A variable of order $k$ is treated like a leaf symbol which is accepted by all states of the same order.
Let $P \subseteq Q$ be a set of states, 
and let $t : D \to (\Sigma\cup\Vars)$ be an input tree. 
A \emph{run tree from $P$ on $t$} is a $\pow Q$-tree%
\footnote{Strictly speaking $\pow Q$ does not have a rank/order.
It is easy to duplicate each subset at every rank/order to obtain an ordered alphabet,
which we avoid for simplicity.}
$s : D \to \pow Q$ over the same tree domain $D$ 
%
s.t.{\@} $s(\varepsilon) = P$, and: 
	%
	%
	%
		\begin{inparaenum}
			\item[i)] if $t(u) = a$ is not a variable and of rank $n$, 
			then 
			$s(u) \trans a s(u\cdot 1) \cdots s(u\cdot n)$, and
			\item[ii)] if $t(u) = \var x$ then $\forall p \in s(u), \ord (p) = \ord (\var x)$.
			%
		\end{inparaenum}
%
%
%
%
%
%
The \emph{language} recognized by a set of states $P \subseteq Q$, 
denoted by $\lang P$, is the set of $\Sigma$-trees $t$
s.t.{\@} there exists a run tree from $P$ on $t$.


%


\section{Ordered tree-pushdown systems}

\label{sec:otpds}

We introduce a generalization of pushdown systems,
where the pushdown is a tree instead of a word.
An \emph{alternating ordered tree-pushdown system} (AOTPS) of order $n \in \Natpos$ is a tuple
$\Sys = \tuple {n, \Sigma, P, \Rules}$
where $\Sigma$ is an ordered alphabet containing symbols of order at most $n$, 
$P$ is a finite set of \emph{control locations},
and $\Rules$ is a set of rules of the form $p, l \rewritesto S, r$
s.t.{\@} $p \in P$ and $S \subseteq P$.
Moreover, $l \rewritesto r$ is a rewrite rule over
$\Sigma$ of one of the two forms:
\begin{equation*}
  \text{\emph{(shallow)}: }a(u_1,\dots,u_m)\rewritesto r\quad
\text{or}\quad
\text{\emph{(deep)}: }a(u_1,\dots,u_k,b(v_1,\dots,v_{m'}),u_{k+1},\dots,u_m)\rewritesto r 
\end{equation*}
where each $u_i,v_j$ is either
  $r$-ground or a variable, 
and for the second form we require
  \begin{equation*}
    \text{\emph{(ordering condition)}: } \text{if
      $\ord(u_i)\leq\ord(b)$, then $u_i$ is  $r$-ground;  for
  $i=1,\dots,m$. }
  \end{equation*}
The rules in $\Rules$ where $l\rewritesto r$ is of the first form are called \emph{shallow}, the others are \emph{deep}.
The tree $b(v_1,\dots,v_{m'})$ in a deep rule is called the \emph{lookahead subtree} of $l$.
A rule $l\to r$ is \emph{flat}\label{def:flat} if each $u_i,v_j$ is just a variable.
%
Let $\Rules_{\ord (b)}$ be the set of deep rules, where the lookahead symbol $b$ is of order $\ord (b)$.
%
For example, $a(\var x, \var y) \rewritesto c(a(\var x, \var y), \var x)$ is shallow and flat,
but $a(b(\var x), \var y) \rewritesto c(\var x, \var y)$ is deep (and flat); here
necessarily  $\ord (\var y) > \ord (b)$.
Finally, $a(c,d, \var x) \rewritesto b(\var x)$ is not flat since $c$ and $d$ are not variables.
In Sec.~\ref{sec:examples} we provide more examples of such rewrite rules by encoding many popular formalisms.
%
%
%
%
While $l$ must be linear, $r$ may be non-linear, thus sub-trees can be duplicated.
The \emph{size} of $\Sys$ is $\size \Sys := \size \Sigma + \size P + \size \Rules$,
where $\size \Rules := \sum_{(p, l \rewritesto S, r) \in \Rules} (1 + \size l + \size S + \size r)$.
%
%

%
%
%
%
Rewrite rules induce an alternating transition system $\tuple {\mathcal C_\Sys, \rewritesto_\Sys}$ by root rewriting.
The set of configurations $\mathcal C_\Sys$ consists of pairs $(p, t)$ with $p \in P$ and $t \in \Trees \Sigma$,
and, for every configuration $(p, t)$, set of control locations $S\subseteq P$, and tree $u$,
$(p, t) \rewritesto_\Sys S \times \set {u}$
if 
there exists a rule $((p, l) \rewritesto (S, r)) \in \Rules$
and a substitution $\sigma$ s.t.{\@} $t = l\sigma$ and $u = r\sigma$. 
%

Let $\ATA A = \tuple {\Sigma, Q, \Delta}$ be a tree automaton s.t.{\@} $P \subseteq Q$.
The \emph{language of configurations} recognized by $\ATA A$ from $P$
is $\lang {\ATA A, P} := \setof {(p, t) \in \mathcal C} {p \in P \textrm { and } t \in \lang p}$.
Given an initial configuration $(p_0, t_0) \in \mathcal C$ and a tree automaton $\ATA A$
recognizing a regular set of target configurations $\lang{\ATA A, P} \subseteq \mathcal C$,
the \emph{reachability problem} for $\Sys$ amounts to determining
whether $(p_0, t_0) \in \prestar {\lang{\ATA A, P}}$.
%
%
%

\ignore{
\begin{align}
	\tag*{$(shallow)$}
	\label{eq:shallow:rule}
	l := \begin{tikzpicture}[baseline=-2.3ex,scale=0.5,level/.style={sibling distance = 1.5cm}]
		\node {$a$}
			child{ node {$v_1$}}
		    child{ node {$\cdots$} edge from parent[draw=none]}
			child{ node {$v_n$}};
	\end{tikzpicture}
	&\longgoesto
	r:= \begin{tikzpicture}[baseline=-2.3ex,scale=0.5,level/.style={sibling distance = 1.5cm}]
		\node {$b$}
			child{ node {$t_1$}}
		    child{ node {$\cdots$} edge from parent[draw=none]}
			child{ node {$t_h$}};
	\end{tikzpicture}
	\\
	\tag*{$(deep)$}
	\label{eq:deep:rule}
	l := \begin{tikzpicture}[baseline=-2.3ex,scale=0.5,level/.style={sibling distance = 1.5cm}]
		\node {$b$}
			child{ node {$u_1$}}
		    child{ node {$\cdots$} edge from parent[draw=none]}
			child{ node {$a$}
				child { node {$v_1$}}
				child { node {$\cdots$} edge from parent[draw=none]}
				child { node {$v_n$}}
				}
		    child{ node {$\cdots$} edge from parent[draw=none]}
		    child{ node {$u_m$}};
	\end{tikzpicture}
	&\longgoesto
	r:= \begin{tikzpicture}[baseline=-2.3ex,scale=0.5,level/.style={sibling distance = 1.5cm}]
		\node {$c$}
			child{ node {$t_1$}}
		    child{ node {$\cdots$} edge from parent[draw=none]}
			child{ node {$t_h$}};
	\end{tikzpicture}
\end{align}
s.t.{\@} $l$ is $r$-shallow, and
the \ref{eq:scan:rule} rule is
%
\begin{enumerate}[(I)]
	
	\item \label{enum:I}
	For any $v_i \in \set{v_1, \dots, v_n}$,
	either $v_i$ is a variable,
	or it is semi-ground (i.e., none of its variables appear on the right).
	
	\item \label{enum:II}
	For any $u_j \in \set{u_1, \dots, u_m}$:
	\begin{enumerate}
		\item If $\ord(u_j) \leq \ord(a)$,
		then $u_j$ is semi-ground.
		
		\item If $\ord(u_j) > \ord(a)$,
		then either $u_j = \var x$ is a variable,
		or it is semi-ground.
	\end{enumerate}
	Intuitively, in the \ref{eq:scan:rule} rule
	subtrees of order not greater than the scanned subtree (the one with root $b$) shall be discarded.
	In other words, $u_j$ is either semi-ground, 
	or a variable of order strictly higher than $b$.
\end{enumerate}
%

%

\begin{remark}
	\label{rem:pop}
	The \ref{eq:pop:rule} rule can be simulated by the \ref{eq:scan:rule} rule
	by replacing $\var x$ with $b(\var x_1, \dots, \var x_n)$ for every $b \in \Sigma$. 
	We keep pop rules for simplicity of exposition.
	\ignore{
	For every \ref{eq:pop:rule} rule above and for every $b \in \Sigma$ of rank $n$ we have the following scan rule:
	%
	\begin{align*}
		\tag*{$(\scanpop)$}
		\label{eq:scan-pop:rule}
		p,\ l := \begin{tikzpicture}[baseline=-2.3ex,scale=0.5,level/.style={sibling distance = 1.5cm}]
			\node {$a$}
				child{ node {$u_1$}}
			    child{ node {$\cdots$} edge from parent[draw=none]}
				child{ node {$b$}
					child { node {$\varord y 1 {}$}}
					child { node {$\cdots$} edge from parent[draw=none]}
					child { node {$\varord y n {}$}}
					}
			    child{ node {$\cdots$} edge from parent[draw=none]}
			    child{ node {$u_m$}};
		\end{tikzpicture}
		&\longgoesto
		P,\ r:= \begin{tikzpicture}[baseline=-2.3ex,scale=0.5,level/.style={sibling distance = 1.5cm}]
			\node {$b$}
				child{ node {$\varord y 1 {}$}}
			    child{ node {$\cdots$} edge from parent[draw=none]}
				child{ node {$\varord y n {}$}};
		\end{tikzpicture}
	\end{align*}
	Note that this is a valid scan rule,
	since the $u_i$'s are semi-ground by construction.
	%
	}
\end{remark}

}

\ignore{

\begin{remark}
	\label{rem:states}
	The \ref{eq:scan:rule} rule allows to remove control locations by encoding them in the root of the tree.
	More precisely, given an AOTPS $\Sys = \tuple {\Sigma, Q, \Rules}$
	we can construct an equivalent \lorenzo{make this formal}
	AOTPS without states $\Sys' = \tuple {\Sigma', Q', \Rules'}$
	where $\Sigma' = Q \times \Sigma$ is an ordered alphabet s.t.{\@}
	$(p, a)$ in $\Sigma'$ has the same rank and order as $a$ in $\Sigma$.
	A configuration $p, a(t_1, \dots, t_n)$ is thus represented as $(p,a)(t_1, \dots, t_n)$.
	We assume that $\Sys$ does not contain pop rules, by the previous Remark~\ref{rem:pop}.
	The set of rules $\Rules'$ is defined as follows.
	If $\Rules$ contains a \ref{eq:push:rule} or \ref{eq:scan:rule} rule as above,
	then $\Rules'$ contains a corresponding rule which is obtained by just replacing the root $a$ with $(p, a)$,
	and similarly with $b, c$:
	\begin{align}
		\tag*{$(\pushp)$}
		\label{eq:push':rule}
		l' := \begin{tikzpicture}[baseline=-2.3ex,scale=0.5,level/.style={sibling distance = 1.5cm}]
			\node {$(p, a)$}
				child{ node {$v_1$}}
			    child{ node {$\cdots$} edge from parent[draw=none]}
				child{ node {$v_n$}};
		\end{tikzpicture}
		&\longgoesto
		r' := \begin{tikzpicture}[baseline=-2.3ex,scale=0.5,level/.style={sibling distance = 1.5cm}]
			\node {$(q, b)$}
				child{ node {$t_1$}}
			    child{ node {$\cdots$} edge from parent[draw=none]}
				child{ node {$t_h$}};
		\end{tikzpicture}
	\end{align}
	\begin{align}
		\tag*{$(\scanp)$}
		\label{eq:scan':rule}
		l' := \begin{tikzpicture}[baseline=-2.3ex,scale=0.5,level/.style={sibling distance = 1.5cm}]
			\node {$(p, a)$}
				child{ node {$u_1$}}
			    child{ node {$\cdots$} edge from parent[draw=none]}
				child{ node {$b$}
					child { node {$v_1$}}
					child { node {$\cdots$} edge from parent[draw=none]}
					child { node {$v_n$}}
					}
			    child{ node {$\cdots$} edge from parent[draw=none]}
			    child{ node {$u_m$}};
		\end{tikzpicture}
		&\longgoesto
		r' := \begin{tikzpicture}[baseline=-2.3ex,scale=0.5,level/.style={sibling distance = 1.5cm}]
			\node {$(q, c)$}
				child{ node {$t_1$}}
			    child{ node {$\cdots$} edge from parent[draw=none]}
				child{ node {$t_h$}};
		\end{tikzpicture}
	\end{align}
\end{remark}
}

\subsection{Reachability analysis}
\label{sec:saturation}

We present a saturation-based procedure to decide reachability in
AOTPSs. This also shows that backward reachability relation preserves
regularity. 
%


\begin{theorem}[Preservation of recognizability]\label{thm:main}
	Let $\Sys$ be an order-$n$ AOTPS
	and let $C$ be regular set of configurations.
	Then, $\prestar{C}$ is effectively regular,
	and an automaton recognizing it can be built in $n$-fold exponential time.
\end{theorem}
\noindent
Let $\Sys = \tuple {n, \Sigma, P, \Rules}$ be an AOTPS.
The target set $C$ is given as a tree automaton $\ATA A = \tuple {\Sigma, Q, \Delta}$ s.t.{\@} $\lang{\ATA A, P} = C$.
W.l.o.g.{\@} we assume that in $\ATA A$ initial states (states in $P$) have no incoming transitions.
Classical saturation algorithms for pushdown automata proceed by adding transitions to the original automaton $\ATA A$,
until no more new transitions can be added.
Here, due to the lookahead of the l.h.s.{\@} of deep rules,
we need to also add new states to the automaton.
However, the total number of new states is bounded once the order of the AOTPS is fixed,
which guarantees termination.
We construct a tree automaton 
$\ATA B = \tuple {\Sigma, Q', \Delta'}$
recognizing $\prestar{\lang {\ATA A, P}}$,
where $Q'$ is obtained by adding states to $Q$,
and $\Delta'$ by adding transitions to $\Delta$,
according to a saturation procedure described below.

%
\ignore{
First, 
we add a universal state $p^*$ accepting all trees. 
Formally, let $Q_0 = \set {p^*}$,
and let $\Delta_0$ contain, for every symbol $a$ of rank $n$, 
the following transition
\begin{align}
	\tag*{($\Delta'$-univ)}
	\label{eq:delta':univ}
	p^* \trans a p^* \cdots p^* \in \Delta_0
\end{align}
%
}%
For every rule $(p, l \rewritesto S, r) \in \Rules$ and for every subtree $v$ of $l$
we create a new state $p^v$ of the same order as $v$ recognizing all $\Sigma$-trees\ignore{of order $k$} that can be obtained
by replacing variables in $v$ by arbitrary trees,
i.e., $\lang {p^v} = \setof {v\sigma} {\sigma : \Vars \to \Trees
  \Sigma,\ v\sigma\in\Trees\Sigma}$; recall that the substitution
should respect the order.
Let $Q_0$ be the set of such $p^v$'s, and let $\Delta_0$ contain the required transitions. 
Notice that $\size {Q_0}, \size {\Delta_0} \leq \size {\mathcal R}$.
%

\label{def:B}
In order to deal with deep rules 
we add new states in the following stratified way.
Let $Q'_{n+1}=Q\cup Q_0$. We define sets
$Q'_{n},\dots,Q'_1$ inductively starting with $Q'_{n}$. Assume
that $Q'_{i+1}$ is already defined. We make $Q'_i$ contain $Q'_{i+1}$.
Then we add to $Q'_i$ states for every deep rule $g\in \Rules_i$ of the form
$p,a(u_1,\dots,u_k,b(\dots),u_{k+1},\dots,u_m)\rewritesto S,r$, with
$\ord(b)=i$. For simplicity of notation, let us suppose that 
$u_1,\dots,u_k$ are of order at most $\ord(b)$, and that
$u_{k+1},\dots,u_m$ are of order strictly greater than $\ord(b)$%
\footnote{This assumption is w.l.o.g.{\@} since one can always add shallow rules to reorder subtrees and put them in the required form.}.
We add to
$Q'_i$ states:
\begin{equation*}
  (g,P_{k+1},\dots,P_m)\in Q'_i\qquad\text{for all $P_{k+1},\dots,P_m\subseteq Q'_{i+1}$.}
\end{equation*}
In particular, to $Q_n$ we add states of the form $(g)$ since $n$ is the
maximal order. We define the set of states in $\ATA B$ to be $Q' := Q'_1$.


%
We add transitions to $\ATA B$ in an iterative process until no more transitions can be added.
During the saturation process, we maintain the following invariant:
\emph{For $1 \leq i \leq n$, states in $Q'_i\setminus Q'_{i+1}$ recognize only trees of order $i$}.
Therefore, $\ATA B$ is also an ordered tree automaton.  Formally,
$\Delta'$ is the least set containing $\Delta \cup \Delta_0$ and
closed under adding transitions according to the following procedure.
Take a deep rule
\begin{equation*}
	g=(p,a(u_1,\dots,u_k,b(v_1,\dots,v_{m'}),u_{k+1},\dots,u_m)\rewritesto S,r) \in \Rules_{\ord(b)}
\end{equation*}
and assume as before that the order of $u_j$ is at most $\ord(b)$ for $j\leq
k$, and strictly bigger than $\ord(b)$ otherwise. We
consider a run tree $t$ from $S$ on $r$ in $\ATA B$. For every
$j=1,\dots,m$ we set:
$P^t_j=\set{p^{u_j}}$ if $u_j$ is $r$-ground, and
$P^t_j= \bigcup t(r^{-1}(\var x))$ if
$u_j=\var x$ is a variable appearing in $r$.
The set $\bigcup t(r^{-1}(\var x))$ collects all states of $\ATA B$ from which the subtree for which $\var x$ can be replaced must be accepted.
%
Moreover, for the lookahead 
subtree $b(v_1, \dots, v_{m'})$,
we let $P_b^t = \set{(g, P^t_{k+1},\dots,P^t_m)}$.
Analogously, we define $S^t_1,\dots,S^t_{m'}$ considering
$v_1,\dots,v_{m'}$ instead of $u_1,\dots,u_m$. Then, we add two transitions:
\begin{equation}\label{eq:rules}
  	p \trans a P^t_1 \cdots P^t_k P_b^t P_{k+1}^t\cdots P^t_m\qquad\text{and}\qquad
        (g,P^t_{k+1},\dots,P^t_m)\trans b S^t_1\dots S^t_{m'}\ .
\end{equation}
Thanks to the ordering condition, $P^t_{k+1},\dots,P^t_m\subseteq
Q'_{\ord(b)+1}$, so $(g,P^t_{k+1},\dots,P^t_m)$ is indeed a state in
$Q'_{\ord(b)}$. For a shallow rule $g$ the procedure is the same but
ignoring the part about the $b(v_1,\dots,v_{m'})$ component; so only one rule
is added in this case.

\theoremstyle{plain}
\newtheorem{lem}[theorem]{Lemma}
\begin{restatable}[Correctness of saturation]{lem}{lemsaturationiscorrect}
	\label{lem:saturation:correct}
	For $\ATA A$ and $\ATA B$ be as above,
	 $\lang {\ATA B, P} = \prestar {\lang {\ATA A,P}}$.
\end{restatable}
\noindent
The correctness proof, even though short, is presented in App.~\ref{app:correctness}. 
The right-in-left inclusion is by straightforward induction on the number of rewrite steps to reach $\lang{\ATA A, P}$.
The left-in-right inclusion is more subtle, but with an appropriate
invariant of the saturation process it also follows by a direct inspection. 



\subsection{Complexity}

The reachability problem for AOTPSs can be solved using the saturation
procedure from Theorem~\ref{thm:main}. 
For an initial configuration $(p_0, t_0) \in \mathcal C$
and an automaton $\ATA A$ recognizing a regular set of target configurations $\lang{\ATA A, P}$,
we construct  $\ATA B$ as in the previous section,
and then test $(p_0, t_0) \in \lang {\ATA B, P}$. In this section we
will analyze the complexity of this procedure in several relevant cases.
All lower-bounds follow from the reductions presented in Sec.~\ref{sec:examples}.

Let $m > 1$ be the maximal rank of any symbol in $\Sigma$.
Using the notation from the previous subsection,
we have that $\size {Q'_{n+1}} \leq \size Q + \size {\Rules}$,
$\size{Q'_n} \leq \size {Q'_{n+1}} + \size \Rules$, and for every $k\in\{1,\dots,n-1\}$,
$\size {Q'_k} \leq \size {Q'_{k+1}} + \size {\Rules} \cdot 2^{(m-1) \cdot \size {Q'_{k+1}}}
		\leq O\left(\size {\Rules} \cdot 2^{(m - 1) \cdot \size {Q'_{k+1}}}\right)$,
and thus
%
	$\size{Q'} \leq \nexp {n-1}{O((m - 1) \cdot (\size Q + \size {\Rules}))}$,
%
where $\nexp 0 x = x$ and, for $i \geq 0$, $\nexp {i+1} x = 2^{\nexp i x}$.
The size of the transition relation is at most one exponential more than the number of states,
thus $\size{\Delta'} \leq \nexp {n}{O((m - 1) \cdot (\size Q + \size {\Rules}))}$. This implies:
%
	%
	%
	%
%
\begin{theorem}
	\label{thm:complexity}
	Reachability in order-$n$ AOTPSs is \nexptimec n.
\end{theorem}

%
%
\noindent
We identify four subclasses of AOTPSs,
for which the reachability problem is of progressively decreasing complexity.
First, we can save one exponential if we consider control-state
reachability for the class of \emph{non-deterministic, flat} AOTPSs.
A system is \emph{non-deterministic} when for every rule $p, l \rewritesto S, r$, the
set $S$ is a singleton.
A system is \emph{flat} when its rules $p, l \rewritesto S, r$ are flat (defined on page \pageref{def:flat}).
%
\emph{Control-state reachability} of a given set of locations $T\subseteq P$ means that the
language of final configurations is $T\times \Trees \Sigma$.
A proof of the theorem below is presented in App.~\ref{app:complexity:nondeterministic}. 
\begin{restatable}{theorem}{thmcomplexitynondeterministic}
	\label{thm:complexity:nondeterministic}
	Control-state reachability in order-$n$ non-deterministic flat AOTPSs is \nexptimec {(n-1)}, where $n \geq 2$.
\end{restatable}

\noindent
Second, we consider the class of \emph{linear} non-deterministic systems.
Suppose that we consider \emph{non-deterministic reachability}, i.e., that $\ATA A$ is non-deterministic. 
When $\Sys$ is linear, i.e., variables in the r.h.s.{\@} of rules in $\Rules$ appear exactly once,
then all $P^t_i$'s and $S^t_i$'s in~\eqref{eq:rules} are singletons,
and thus $\ATA B$ is also non-deterministic.
Consequently, the only states from $Q_i'\setminus Q'_{i+1}$ that are
used by rewriting 
rules have the form $(g,\set{p_{k+1}},\dots,\set{p_m})$ for $p_{k+1},\dots,p_m\in
Q'_{i+1}$. Therefore, there are at most 
$O (\left(\size Q+\size {\Rules}\right)^{(m-1)^{n}})$ states
and $O(\size {\Rules} \cdot {\size {Q'}}^m)$ transitions,
and $\ATA B$ is thus doubly exponential in $n$. 
	%
	%
%
\begin{theorem}
	\label{thm:complexity:linear:nondet}
	Non-deterministic reachability in linear non-deterministic AOTPSs is \twoexptimec.
\end{theorem}
%
%
The next simplification is when the system is \emph{shallow} in the
sense that it does not have deep rules. In this case we do not need to
add states recursively ($Q':= Q \cup Q_0$), and we thus avoid the multiple exponential blow-up.
Similarly, when the system is \emph{unary}, i.e., the maximal rank is $m=1$,
only polynomially many states are added.
\begin{theorem}
	\label{thm:complexity:shallow:unary}
	Reachability in shallow as well as in unary AOTPSs is \exptimec.
\end{theorem}
If moreover the system is non-deterministic, then we get \ptime\
complexity, provided the rank of the letters in the alphabet is bounded.
\begin{theorem}
	\label{thm:complexity:shallow:unary:nondeterministic}
	Non-deterministic reachability in unary non-deterministic AOTPSs and in shallow non-deterministic AOTPSs of fixed rank is in \ptime.
\end{theorem}
%

%




\subsection{Expressiveness}
\label{sec:expressiveness}

In the next section we give a number of examples of systems that can
be directly encoded in AOTPSs. Before that, we would like to underline that
AOTPSs can themselves be encoded into collapsible pushdown systems. 
We formally formulate this equivalence in terms of Krivine machines with states,
which are defined later in Sec.~\ref{sec:krivine:machine}.
The details of this reduction are presented in App.~\ref{app:aotps2km}.

\begin{restatable}{theorem}{thmaotpstokm}
	\label{thm:aotps2km}
	Every AOTPS of order $n$ can be encoded in a Krivine machine with states of the same level
	s.t.{\@} every rewriting step of the AOTPS corresponds to a number of reduction steps of the Krivine machine.
\end{restatable}

\noindent
Since parity games over the configuration graph of the Krivine machine with states are known to be decidable \cite{SalvatiWalukiewicz:InfComp:2014},
this equivalence yields decidability of parity games over AOTPSs.
However, in this paper we concentrate on reachability properties of AOTPSs,
which are decidable thanks to our simple saturation algorithm from Sec.~\ref{sec:saturation}.
No such saturation algorithm was previously known for the Krivine machine with states.

\section{Applications}

\label{sec:examples}

In this section, we give several examples of systems that can be encoded as AOTPSs.
Ordinary alternating pushdown systems (and even prefix-rewrite systems) can be easily encoded as \emph{unary} AOTPSs
by viewing a word as a linear tree; the ordering condition is trivial since symbols have rank $\leq 1$.
Moreover, \emph{tree-pushdown systems} \cite{Guessarian:TPDA:1983} can be seen as \emph{shallow} AOTPSs.
By Theorem~\ref{thm:complexity:shallow:unary}, reachability is in \exptime for both classes,
and, by Theorem~\ref{thm:complexity:shallow:unary:nondeterministic},
it reduces to \ptime for the non-alternating variant (for fixed maximal rank).
%
%
%

In the rest of the section, we show how to encode four more sophisticated classes of systems,
namely \emph{ordered multi-pushdown systems} (Sec.~\ref{sec:ordered:mpds}),
\emph{annotated higher-order pushdown systems} (Sec.~\ref{sec:annotated:pda}),
the \emph{Krivine machine with states} (Sec.~\ref{sec:krivine:machine}),
and \emph{ordered annotated multi-pushdown systems} (Sec.~\ref{sec:oapds}),
and we show that reachability for these models (except the last one) can be decided with tight complexity bounds using our conceptually simple saturation procedure.
%
%
%


\subsection{Ordered multi-pushdown systems}

\label{sec:ordered:mpds}

In an ordered multi-pushdown system there are $n$ pushdowns.
Symbols can be pushed on any pushdown,
but only the first non-empty pushdown can be popped
\cite{BreveglieriCherubiniCitriniCrespi-Reghizzi:Ordered:1996,AtigBolligHabermehl:Ordered:2008,Atig:ordered:2012}.
This is equivalent to saying that to pop a symbol from the $k$-th pushdown,
the contents of the previous pushdowns $1, \dots, k-1$ should be discarded.
Formally, an \emph{alternating ordered multi-pushdown system} is a tuple $\OMPDS = \tuple {n, \Gamma, Q, \Delta}$,
where $n \in \Natpos$ is the \emph{order} of the system (i.e., the number of pushdowns),
$\Gamma$ is a finite pushdown alphabet,
$Q$ is a finite set of control locations,
and $\Delta \subseteq Q \times O_n \times 2^Q$
is a set of rules of the form $(p, o, P)$
with $p \in Q$, $P \subseteq Q$,
and $o$ a pushdown operation in
$O_n := \setof {\opush k a, \opop k a} {1 \leq k \leq n, a \in \Gamma}$.
We say that $\OMPDS$ is \emph{non-deterministic} when $P$ is a singleton for every rule.
A multi-pushdown system induces an alternating transition system $\tuple {\mathcal C_\OMPDS, \goesto_\OMPDS}$
where the set of configurations is $\mathcal C_\OMPDS = Q \times (\Gamma^*)^n$,
and the transitions are defined as follows: 
for every $(p, \opush k a, P) \in \Delta$ there exists a transition $(p, w_1, \dots, w_n) \goesto_\OMPDS P \times \set{(w_1, \dots, a \cdot w_k, \dots, w_n)}$,
and for every $(p, \opop k a, P) \in \Delta$ there exists a transition $(p, w_1, \dots, a \cdot w_k, \dots, w_n) \goesto_\OMPDS P \times \set {(\varepsilon, \dots, \varepsilon, w_k, \cdots, w_n)}$.
%
%
For $c\in\mathcal C_\OMPDS$ and $T \subseteq Q$,
the \emph{(control-state) reachability problem} for $\OMPDS$ asks whether
$c \in \prestar {T \times (\Gamma^*)^n}$.

\paragraph{Encoding.}

We show that an ordered multi-pushdown system can be simulated by an AOTPS.
The idea is to encode the $k$-th pushdown as a linear tree of order $k$,
and to encode a multi-pushdown as a tree of linear pushdowns.
Let $\bot$ and $\bullet$ be two new symbols not in $\Gamma$,
let $\Gamma_\bot = \Gamma \cup \set \bot$,
and let $\Sigma = (\Gamma_\bot \times \set{1,\dots,n})\cup\set{\bullet}$ be an ordered alphabet,
where a symbol $(a, i) \in \Gamma_\bot \times \set i$ has order $i$,
rank $1$ if $a \in \Gamma$
and rank $0$ if $a = \bot$.
Moreover, $\bullet$ has rank $n$ and order $1$.
For simplicity, we write $a^i$ instead of $(a, i)$.
%
%
%
A multi-pushdown $w_1, \dots, w_n$, where each $w_j = a_{j, 1} \dots a_{j, n_j}$
is encoded as the tree
$\enc {} {w_1, \dots, w_n} := \bullet(a^1_{1, 1}(a^1_{1, 2}(\dots \bot^1)), \dots, a^n_{n, 1}(a^n_{n, 2}(\dots \bot^n)))$.
\ignore{
\begin{align*}
	\enc {} {w_1, \dots, w_n} =
	\begin{tikzpicture}[baseline=-2.3ex,scale=0.6,level/.style={sibling distance = 2.5cm}]
		\node {$\bullet$}
			child{ node {$a^1_{1, 1}$} child {node {$a^1_{1, 2}$} child {node {$\vdots$} child {node {$\bot^1$}}}}}
			child{ node {$a^2_{2, 1}$} child {node {$a^2_{2, 2}$} child {node {$\vdots$} child {node {$\bot^2$}}}}}
		    child{ node {$\cdots$} edge from parent[draw=none]}
			child{ node {$a^n_{n, 1}$} child {node {$a^n_{n, 2}$} child {node {$\vdots$} child {node {$\bot^n$}}}}};
	\end{tikzpicture}
\end{align*}
}
%
%
For an ordered multi-pushdown system $\OMPDS = \tuple {n, \Gamma, Q, \Delta}$
we define an equivalent AOTPS $\Sys = \tuple {n, \Sigma, Q, \Rules}$
with $\Sigma$ defined as above,
and set of rules $\Rules$ defined as follows (we use the convention that variable $\var x_k$ has order $k$):
For every push rule $(p, \opush k a, P) \in \Delta$,
we have a rule $(p, \bullet(\var x_1, \dots, \var x_n) \goesto P, \bullet(\var x_1, \dots, a^k(\var x_k), \dots, \var x_n))\in \Rules$,
\ignore{
\begin{align*}
	p, \begin{tikzpicture}[baseline=-2.3ex,scale=0.5,level/.style={sibling distance = 1.4cm}]
		\node {$\bullet$}
			child{ node {$\var x_1$}}
		    child{ node {$\cdots$} edge from parent[draw=none]}
			child{ node {$\var x_n$}};
	\end{tikzpicture}
	\longgoesto
	P, \begin{tikzpicture}[baseline=-2.3ex,scale=0.5,level/.style={sibling distance = 1.4cm}]
		\node {$\bullet$}
			child{ node {$\var x_1$}}
		    child{ node {$\cdots$} edge from parent[draw=none]}
			child{ node {$a^k$} child {node {$\var x_k$}}}
		    child{ node {$\cdots$} edge from parent[draw=none]}		
			child{ node {$\var x_n$}};
	\end{tikzpicture}
\end{align*}
}%
%
and for every pop rule $(p, \opop k a, P) \in \Delta$, we have
$(p, \bullet(\var x_1, \dots, a^k(\var x_k), \dots, \var x_n) \goesto P, \bullet(\bot^1, \dots, \bot^{k-1}, \var x_k, \var x_{k+1}, \dots, \var x_n)) \in \Rules$.
\ignore{
\begin{align*}
	p, \begin{tikzpicture}[baseline=-2.3ex,scale=0.7,level/.style={sibling distance = 1cm}]
		\node {$\bullet$}
			child{ node {$\varord x 1 {}$}}
		    child{ node {$\cdots$} edge from parent[draw=none]}
			child{ node {$\varord x {k-1} {}$}}
			child{ node {$a^k$} child {node {$\varord x k {}$}}}
			child{ node {$\varord x {k+1} {}$}}
		    child{ node {$\cdots$} edge from parent[draw=none]}
			child{ node {$\varord x n {}$}};
	\end{tikzpicture}
	\longgoesto
	P, \begin{tikzpicture}[baseline=-2.3ex,scale=0.7,level/.style={sibling distance = 1cm}]
		\node {$\bullet$}
			child{ node {$\bot^1$}}
		    child{ node {$\cdots$} edge from parent[draw=none]}
			child{ node {$\bot^{k-1}$}}
			child{ node {$\varord x k {}$}}
			child{ node {$\varord x {k+1} {}$}}
		    child{ node {$\cdots$} edge from parent[draw=none]}		
			child{ node {$\varord x {n} {}$}};
	\end{tikzpicture}
\end{align*}
where variable $\varord x i {}$ has order $i$.
This is a valid scan rule since low order variables $\varord x i {}$'s with $i < k$ do not appear on the right-hand side  (and thus semi-ground w.r.t. this rewrite rule).
}
Both kinds of rules above are linear, and the latter one satisfies the ordering condition
since lower-order variables $\var x_1, \dots, \var x_{k-1}$ are discarded.
%
%
It is easy to see that $(p, w_1, \dots, w_n) \goesto_\OMPDS^* P \times \set{(w'_1, \dots, w'_n)}$ if, and only if,
$(p, \enc {} {w_1, \dots, w_n}) \goesto_{\Sys}^* P \times \set{\enc {} {w'_1, \dots, w'_n}}$.
Thus, the encoding preserves reachability properties.
%
%
By Theorem~\ref{thm:complexity}, we obtain an \nexptime n upper-bound for reachability in alternating multi-pushdown systems of order $n$.
Moreover, since $\Sys$ is linear,
and since $\Sys$ is non-deterministic when $\OMPDS$ is non-deterministic,
by Theorem~\ref{thm:complexity:linear:nondet}
we recover the optimal \twoexptimec complexity proved by \cite{AtigBolligHabermehl:Ordered:2008} (cf.{\@} also \cite{Atig:ordered:2012}).
\begin{theorem}[\cite{AtigBolligHabermehl:Ordered:2008}]
	Reachability in alternating ordered multi-pushdown systems is in \nexptime n.
	Reachability in non-deterministic ordered multi-pushdown systems is \twoexptimec.
\end{theorem}
\noindent
Reachability for the alternating variant of the model (in \nexptime n) was not previously known.


\ignore{

\subsection{Saturation rules}

Let $n$ be the number of pushdowns.
Let $\ATA A = \tuple {\set p \cup Q, \Sigma, \Delta}$ be a nondeterministic tree automaton recognizing (finite) ordered multi-pushdowns.
There exists a distinguished initial state $p$ recognizing multi-pushdowns,
and internal states in $Q$ recognizing ordinary pushdowns.
The ordering restriction implies that internal states in $Q$ can be ordered w.r.t. a total partial order $\leq$ of height $n$.
Thus $Q$ can be partitioned into $Q = Q_1 \cup \cdots \cup Q_n$,
where $Q_i$ is the $i$-h equivalence class of $\leq$.
Moreover, $q \leq q'$ iff $q \in Q_i$ and $q' \in Q_j$ with $i \leq j$.
Thus, an initial transition in $\ATA A$ takes the form $p \trans {a(k)} q_k \cdots q_n$ with $q_k < \cdots < q_n$.
Moreover, we can assume w.l.o.g.{\@} that internal transitions in $Q$ start and leave from the same class $Q_i$.

We build a new nondeterministic tree automaton $\ATA A'$ recognizing the predecessors of $\lang {\ATA A}$.
The new automaton $\ATA A' = \tuple {\set p \cup Q', \Sigma, \Delta'} $ is obtained by adding states and transitions to $\ATA A$.
The new set of states $Q'$ extends $Q$ and it is partitioned $Q' = Q_1' \cup \cdots \cup Q_n'$,
where the $Q_j'$'s are defined inductively as follows:
For $i \geq 0$,
\begin{align*}
	Q_{n-i}' := Q_{n-i} \cup \set p \times \Sigma \times Q_{n-(i-1)}' \times Q_{n-(i-2)}' \times \cdots \times Q_n'
\end{align*}
The (new) auxiliary states of the form $(p, q_k' \cdots q_n')$ represent \emph{contexts}.

If $Q$ contains $\size Q = O(m)$ states, then
\begin{align*}
	\size {Q'_n} &= \size {Q_n} = O(m) \\
	\size {Q'_{n-1}} &= \size {Q_{n-1}} + \size {Q_n'} = O(m) \\
	\size {Q'_{n-2}} &= \size {Q_{n-2}} + \size {Q_{n-1}'} \cdot \size {Q_n'} = O(m^2) \\
	\vdots \\
	\size {Q'_{n-i}} &= \size {Q_{n-i}} + \size {Q_{n-(i-1)}'} \cdot \size {Q_{n-(i-2)}'} \cdots \size {Q_n'} = O(m^{2^{i-1}})
\end{align*}
Therefore, $\size {Q'_1} = O(m^{2^{n-2}})$.
Thus, it suffices to add doubly-exponential states in $n$.

\begin{align}
	\tag*{$(\pop 1)$}
	\label{rule:ompds:pop1}
	\frac
	{p \trans {b(k)} q_k \cdots q_n}
	{\begin{array}{l}
		p \trans {a(k)} (p, a(k), q_{k+1} \cdots q_n) q_{k+1} \cdots q_n \\
		(p, a(k), q_{k+1} \cdots q_n) \trans b q_k
	\end{array}}
\end{align}
\begin{align}
	\tag*{$(\pop 2)$}
	\label{rule:ompds:pop2}
	\frac
	{p \trans {b(k+1)} q_{k+1} \cdots q_n}
	{\begin{array}{l}
		p \trans {a(k)} q_\bot (p, a(k), q_{k+2} \cdots q_n) q_{k+2} \cdots q_n \\
		(p, a(k), q_{k+2} \cdots q_n) \trans b q_{k+1}
	\end{array}}
\end{align}
\begin{align}
	\tag*{$(\push b {h \leq k})$}
	\label{rule:ompds:push1}
	\frac
	{p \trans {b(h)} q_\bot \cdots q_\bot q_k \cdots q_n \qquad q_k \trans a q_k'}
	{p \trans {a(k)} q_k' q_{k+1} \cdots q_n}
\end{align}
\begin{align}
	\tag*{$(\push b {h > k})$}
	\label{rule:ompds:push2}
	\frac
	{p \trans {a(k)} q_k \cdots q_n \qquad q_h \trans b q_h'}
	{p \trans {a(k)} q_k \cdots q_h' \cdots q_n}
\end{align}

\subsection{Completeness}

\begin{lemma}
	\label{lem:ompds:complete}
	The saturation procedure is complete, i.e., $\prestar {\lang {\ATA A}} \subseteq \lang {\ATA A'}$.
\end{lemma}

\begin{proof}
	By induction on the length of the shortest path to reach $\lang {\ATA A}$.
\end{proof}

\subsection{Soundness}

\begin{align*}
	\sem p &:= \prestar {\lang {\ATA A}} \\
	\sem q &:= \lang q, \textrm{ with } q \in Q \\
	\sem {(p, a(k), q_{k+1} \cdots q_n)} &:= \setof {w_k}
	{\forall w_{k+1} \in \sem {q_{k+1}} \dots \forall w_n \in \sem {q_n}\cdot a(k)(w_k, \dots, w_n) \in \sem p}
\end{align*}

An internal transition $q \trans a q'$ is \emph{sound} if, whenever $w \in \sem {q'}$, $a(w) \in \sem q$.
Similarly, an initial transition $p \trans {a(k)} q_k \cdots q_n$ is \emph{sound} if,
whenever $w_k \in \sem {q_k} \dots w_n \in \sem {q_n}$,
$a(k)(w_k, \dots, w_n) \in \sem p$.
An automaton is sound if all its transitions are sound.

\begin{lemma}
	\label{lem:ompds:sound:incl}
	In a sound automaton, $\lang q \subseteq \sem q$.
\end{lemma}

\begin{proof}
	If $q$ is an internal state, then one proceeds by induction on the length of accepting runs and using the definition of $\sem{\cdot}$.
	If $q$ is the initial state $p$, and $a(k)(w_k, \dots, w_n) \in \lang p$,
	then there is a sound transition $p \trans {a_k} q_k \cdots q_n$ s.t.{\@}
	$w_k \in \lang {q_k} \dots w_n \in \lang {q_n}$.
	By the first part, $w_k \in \sem {q_k} \dots w_n \in \sem {q_n}$.
	Since the transition was sound, we get $a(k)(w_k, \dots, w_n) \in \sem p$.
\end{proof}

\begin{lemma}
	\label{lem:ompds:sound:init}
	The initial automaton $\ATA A$ is sound.
\end{lemma}

\begin{proof}
	Immediate, since the only internal transitions in $\ATA A$ are those between states in $Q$
	for which $\sem q$ coincides with the language recognized by $q$.
\end{proof}

\begin{lemma}
	\label{lem:ompds:sound:induct}
	The saturation process adds only sound transitions.
\end{lemma}

\begin{proof}
	We proceed by case analysis on the newly added rule.
	\paragraph{Case~\ref{rule:ompds:pop1}.}
	Assume that (1) $p \trans {b(k)} q_k \cdots q_n$ is a sound rule.
	We show that (2) $p \trans {a(k)} (p, a(k), q_{k+1} \cdots q_n) q_{k+1} \cdots q_n$
	and (3) $(p, a(k), q_{k+1} \cdots q_n) \trans b q_k$ are sound rules.
	For (2), assume $w_k \in \sem {(p, a(k), q_{k+1} \cdots q_n)}, w_{k+1} \in \sem {q_{k+1}} \dots w_n \in \sem q_n$.
	By the definition of $\sem {(p, a(k), q_{k+1} \cdots q_n)}$,
	$a(k)(w_k, \dots, w_n) \in \sem p$.
	Notice that we didn't use the fact that (1) is sound, i.e., rules of the form (2) are always sound by construction.
	
	For (3), assume $w_k \in \sem {q_k}$, $w_{k+1} \in \sem {q_{k+1}} \dots w_n \in \sem {q_n}$,
	and we have to show that $a(k)(b (w_k), w_{k+1}, \dots, w_n) \in \sem p$.
	Since (1) is sound, we have $b(k)(w_k, \dots, w_n) \in \sem p$.
	But $a(k)(b (w_k), w_{k+1}, \dots, w_n) \goesto {} b(k)(w_k, \dots, w_n)$
	in one step by one application of rule~\ref{rule:ompds:pop1},
	thus $a(k)(b (w_k), w_{k+1}, \dots, w_n) \in \sem p$.
	
	\paragraph{Case~\ref{rule:ompds:pop2}.}
	Similar.
	
	\paragraph{Case~\ref{rule:ompds:push1}.}
	Assume that (1) $p \trans {b(h)} q_\bot \cdots q_\bot q_k \cdots q_n$ and (2) $q_k \trans a q_k'$ are sound rules,
	and we show that (3) $p \trans {a(k)} q_k' q_{k+1} \cdots q_n$ is a sound rule.
	Let $w_k \in \sem {q_k'}, w_{k+1} \in \sem {q_{k+1}} \dots w_n \in \sem {q_n}$,
	and we have to show that $a(k)(w_k, \dots, w_n) \in \sem p$.
	Since (2) is sound, $a(w_k) \in \sem {q_k}$.
	Since (1) is sound, $b(h)(\bot, \dots, \bot, a(w_k), w_{k+1}, \dots, w_n) \in \sem p$.
	But $a(k)(w_k, \dots, w_n) \goesto {} b(h)(\bot, \dots, \bot, a(w_k), w_{k+1}, \dots, w_n)$
	in one step by one application of rule~\ref{rule:ompds:push1},
	thus $a(k)(w_k, \dots, w_n) \in \sem p$.
	
	\paragraph{Case~\ref{rule:ompds:push2}.}
	Similar.	
\end{proof}

\begin{lemma}
	\label{lem:ompds:sound}
	The saturation procedure is sound, i.e., $\lang {\ATA A'} \subseteq \prestar {\lang {\ATA A}}$.
\end{lemma}

\begin{proof}
	By Lemma~\ref{lem:ompds:sound:init}, the initial automaton $\ATA A$ is sound,
	and, by Lemma~\ref{lem:ompds:sound:induct}, the final automaton $\ATA A'$ is sound, too.
	By Lemma~\ref{lem:ompds:sound:incl} applied to the initial state $p$,
	$\lang {\ATA A'} = \lang p \subseteq \sem p = \prestar {\lang {\ATA A}}$.
\end{proof}

}

\subsection{Annotated higher-order pushdown systems}

\label{sec:annotated:pda}

Let $\Gamma$ be a finite pushdown alphabet. In the following, we fix an order $n \geq 1$,
and we let $1 \leq k \leq n$ range over orders.
For our purpose, it is convenient to expose the topmost pushdown at every order recursively.%
\footnote{Our definition is equivalent to \cite{BroadbentCarayolHagueSerre:Saturation:2012}.}
We define $\Gamma_k$, the set of \emph{annotated higher-order pushdowns (stacks) of order $k$},
simultaneously for all $k\in\{1,\dots,n\}$,
as the least set containing the empty pushdown $\stack\,$,
and, whenever $u_1 \in \Gamma_1, \dots, u_k \in \Gamma_k$, $v_j \in \Gamma_j$ for some $j\in\{1,\dots,n\}$,
then $\stack {a^{v_j}, u_1, \dots, u_k} \in \Gamma_k$.
Similarly, if we do not consider stack annotations $v_j$'s,
we obtain the set of \emph{higher-order pushdowns of order $k$}.
Operations on annotated pushdowns are as follows.
The operation $\apush k b$ pushes a symbol $b \in \Gamma$
on the top of the topmost order-$1$ stack
and annotates it with the topmost order-$k$ stack,
%
$\apush k {}$ duplicates the topmost order-$(k-1)$ stack,
$\apop k$ removes the topmost order-$(k-1)$ stack,
%
and $\acollapse k$ replaces the topmost order-$k$ stack with the order-$k$ stack annotating the topmost symbol:
%
%
\vspace{-.5ex}
	\begin{align*}
	  	&\apusharg {k} b {\stack{a^u, u_1, \dots, u_n}}
			= \stack{b^{\stack{a^u, u_1, \dots, u_k}}, \stack{a^u, u_1}, u_2, \dots, u_n}, \\
		&\apusharg k {} {\stack{a^u, u_1, \dots, u_n}}
			= \stack{a^u, u_1, \dots, u_{k-1}, \stack{a^u, u_1, \dots, u_k}, u_{k+1}, \dots, u_n}, \\
		&\apoparg k {\stack{a^u, v_1, \dots,  v_{k-1}, \stack{b^v, u_1, \dots, u_k}, u_{k+1}, \dots, u_n}} = 
			\stack{b^v, u_1, \dots, u_n}, \\
		&\acollapsearg k {\stack{a^{\stack{b^v, v_1, \dots, v_k}}, u_1, \dots, u_n}}
			= \stack{b^v, v_1, \dots, v_k, u_{k+1}, \dots, u_n}.
		\vspace{-4ex}
	\end{align*}
%
Let $O_n = \bigcup_{k=1}^n \setof {\apush k b, \apush k {}, \apop k, \acollapse k} {b \in \Gamma}$
be the set of stack operations.
Similarly, one can define operations $\apush {} b$ and $\apop k$ on stacks without annotations (but not $\acollapse k$, or $\apush k b$).
An \emph{alternating order-$n$ annotated pushdown system} is a tuple $\APDS = \tuple {n, \Gamma, Q, \Delta}$,
where $\Gamma$ is a finite stack alphabet,
$Q$ is a finite set of control locations,
and $\Delta \subseteq Q \times \Gamma \times O_n \times 2^Q$ is a set of rules. 
%
An \emph{alternating order-$n$ pushdown system} (i.e., without annotations) is as $\APDS$ above,
except that we consider non-annotated stack and operations on non-annotated stacks.
An annotated pushdown system induces a transition system $\tuple {\mathcal C_\APDS, \goesto_\APDS}$,
where $\mathcal C_\APDS = Q \times \Gamma_n$,
and the transition relation is defined as
$(p, w) \goesto_\APDS P \times \set {w'}$ whenever $(p, a, o, P) \in \Delta$
with $w = \stack {a^u, \cdots}$ and $w' = o(w)$.
Thus, a rule $(p, a, o, P)$ first checks that the topmost stack symbol is $a$,
and then applies the transformation provided by the stack operation $o$ to the current stack
(which may, or may not, change the topmost stack symbol $a$).
Given $c\in\mathcal C_\APDS$ and $T \subseteq Q$, the \emph{(control-state) reachability problem} for $\APDS$
asks whether $c\in \prestar {T \times \Gamma_n}$.

\paragraph{Encoding.}

We represent annotated pushdowns as trees.
Let $\Sigma$ be the ordered alphabet containing, for each $k\in\{1,\dots,n\}$,
an end-of-stack symbol $\bot^k \in \Sigma$ of rank $0$ and order $k$.
Moreover, for each $a \in \Gamma$ and order $k\in\{1,\dots,n\}$,
there is a symbol $\tuple {a, k}\in\Sigma$ of order $k$ and rank $k+1$
representing the root of a tree encoding a stack of order $k$.
An order-$k$ stack is encoded as a tree recursively by $\enc k {\stack\,} = \bot^k$ and
$\enc k {\stack{a^u, u_1, \dots, u_k}} = \tuple {a, k}(\enc i u, \enc 1 {u_1}, \dots, \enc k {u_k})$,
where $i$ is the order of $u$.
\ignore{
\begin{align*}
	\enc k {\stack{a^u, u_1, \dots, u_k}} &= 
	\begin{tikzpicture}[baseline=-2.3ex,scale=0.6,level/.style={sibling distance = 2.5cm}]
		\node {$a(i, k)$}
			child{ node {$\enc 1 {u_1}$}}
		    child{ node {$\cdots$} edge from parent[draw=none]}
			child{ node {$\enc k {u_k}$}}
			child{ node {$\enc i u$}};
	\end{tikzpicture}
\end{align*}
where $i$ is the order of the stack $u$.
}
Let $\APDS=\tuple {n, \Gamma, Q, \Delta}$ be an annotated pushdown system.
We define an equivalent AOTPS $\Sys = \tuple {n, \Sigma, Q, \Rules}$,
where $\Sigma$ is as defined above,
and $\Rules$ contains a rule $p, l \goesto P, r$ for each rule in $(p, a, o, P) \in \Delta$ and orders $m, m_1$,
where $l \goesto r$ is as follows (cf.{\@} also Fig.~\ref{fig:apds:otpds:rules} in the appendix for a pictorial representation).
We use the convention that a variable subscripted by $i$ has order $i$,
and we write ${\var x}_{i..j}$ for $(\var x_i, \dots, \var x_j)$,
and similarly for $\var z_{i..j}$:
\vspace{-.5ex}
\begin{align*}
	&\tuple{a, n}(\var y_m, \var x_{1..n}) \goesto
		\tuple {b, n}(\tuple{a, k}(\var y_m, \var x_{1..k}), \tuple {a, 1}(\var y_m, \var x_1), \var x_{2..n})
		&\textrm{ if } o = \apush k b,
	\\
	&\tuple{a, n}(\var y_m, \var x_{1..n}) \goesto
		\tuple{a, n}(\var y_m, \var x_{1..k-1}, \tuple{a, k}(\var y_m, \var x_{1..k}), \var x_{k+1..n})
		&\textrm{ if } o = \apush k {},
	\\
	&\tuple{a, n}(\var z'_{m_1}, \var z_{1..k-1}, \tuple{b, k}(\var y_m, \var x_{1..k}), \var x_{k+1..n}) \goesto
		\tuple{b, n}(\var y_m, \var x_{1..n})
		&\textrm{ if } o = \apop k,
	\\
	&\tuple{a, n}(\tuple{b, k}(\var y_m, \var x_{1..k}), \var z_{1..k}, \var x_{k+1..n}) \goesto
		\tuple{b, n}(\var y_m, \var x_{1..n})
		&\textrm{ if } o = \acollapse k.
\end{align*}
%
The last two rules satisfy the ordering condition of AOTPSs
since only higher-order variables $\var x_{k+1}, \dots, \var x_n$ are not discarded.
%
%
	%
%
%
%
%
%
%
It is easy to see that $(p, w) \goesto_\APDS^* P \times \set {w'}$ if, and only if,
$(p, \enc n w) \goesto_{\Sys}^* P \times \set {\enc n {w'}}$.
Consequently, the encoding preserves reachability properties.
Since an annotated pushdown system of order $n$ is simulated by a flat AOTPS of the same order,
the following complexity result is an immediate consequence of Theorems~\ref{thm:complexity} and \ref{thm:complexity:nondeterministic}.
\begin{theorem}[\cite{BroadbentCarayolHagueSerre:Saturation:2012}]
	\label{thm:apds:complexity}
	Reachability in alternating annotated pushdown systems of order $n$ and
	in non-deterministic annotated pushdown systems of order $n+1$ is \nexptimec n.
\end{theorem}

\ignore{

\subsection{Saturation rules}

Let $n$ be the maximal stack order.
Let $\ATA A = \tuple {\set p \cup P, \Sigma, \Delta}$ be an \emph{alternating} tree automaton recognizing (finite) $n$-pushdowns.
There exists a distinguished initial state $p$ recognizing $n$-pushdowns,
and internal states in $P$ partitioned into $P = Q_1 \cup \cdots \cup Q_n$,
where $Q_k$ recognizes $k$-pushdowns, for $1 \leq k \leq n$. 
Moreover, we assume that there exists a distinguished state $q^*$ recognizing every stack.
We define a total pre-order on internal states by letting $q \leq q'$ iff $q \in Q_i$ and $q' \in Q_j$ with $i \leq j$;
let $\approx$ be the induced equivalence relation.
An initial transition in $\ATA A$ takes the form $p \trans {a(n)} q_1 \cdots q_n q$ with $q_1 < \cdots < q_n \approx q$.

We build a new nondeterministic tree automaton $\ATA A'$ recognizing the predecessors of $\lang {\ATA A}$.
The new automaton $\ATA A' = \tuple {\set p \cup P', \Sigma, \Delta'} $ is obtained by adding states and transitions to $\ATA A$.
The new set of states $P'$ extends $P$ and it is similarly partitioned into $Q' = Q_1' \cup \cdots \cup Q_n'$,
where the $Q_j'$'s are defined inductively as follows:
For $i \geq 1$,
\begin{align*}
	Q_n'		&:= Q_n \cup \set p \times \Sigma \times \set {q^*} \cup \set p \times \Sigma \\
	Q_{n-i}'	&:= Q_{n-i} \cup \set p \times \Sigma \times \underbrace{2^{Q_{n-(i-1)}'} \times 2^{Q_{n-(i-2)}'} \times \cdots \times 2^{Q_n'}}_{i} \times \set {q^*}
\end{align*}
For example, $Q_{n-1}' = Q_{n-1} \cup \set p \times \Sigma \times 2^{Q_n'} \times \set {q^*}$.
Notice that $Q_{n-i}$ is well defined since it depends only on $Q_{n-j}'$'s for a smaller $j < i$.
The (new) auxiliary states of the form $(p, a(n), P_k \cdots P_n q^*)$ represent \emph{contexts}.

If $Q$ contains $\size Q = O(m)$ states, then
%
%
$\size {Q'_1}$ is a tower of $n$ exponentials.
TODO: check the calculation.
TODO: why introducing $q^*$ instead of using alternation with $\emptyset$?

\input{apds_saturation_rules}

\subsection{Completeness}

\begin{lemma}
	\label{lem:apds:complete}
	The saturation procedure is complete, i.e., $\prestar {\lang {\ATA A}} \subseteq \lang {\ATA A'}$.
\end{lemma}

\begin{proof}
	By induction on the length of the shortest path to reach $\lang {\ATA A}$.
\end{proof}

\subsection{Soundness}

\begin{align*}
	\sem p &:= \prestar {\lang {\ATA A}} \\
	\sem q &:= \lang q, \textrm{ with } q \in Q \\
	\sem {(p, a(n))} &:= \setof t
	{\forall t_1 \cdots t_n \cdot a(n)(t_1, \dots, t_n, t) \in \sem p} \\
	\sem {(p, a(n), P_{k+1} \cdots P_n q^*)} &:= \setof {t_k}
	{	\forall t_1 \cdots t_{k-1},
		\forall t_{k+1} \in \sem {P_{k+1}} \cdots \forall t_n \in \sem {P_n},
		\forall t \cdot a(n)(t_1, \dots, t_n, t) \in \sem p}
\end{align*}

A(n initial or internal) transition $q \trans {a(k)} P_1 \cdots P_k P$ is \emph{sound} if,
whenever $t_1 \in \sem {P_1} \cdots t_k \in \sem {P_k}, t \in \sem P$,
we have $a(k)(t_1, \dots, t_k, t) \in \sem q$.
An automaton is sound if all its transitions are sound.

\begin{lemma}
	\label{lem:apds:sound:incl}
	In a sound automaton, $\lang q \subseteq \sem q$ for every state $q$.
\end{lemma}

\begin{proof}
	By induction on the length of accepting runs.
	For the base case, let $a \in \lang q$, where $a$ is a leaf (arity 0).
	There exists a sound transition $q \trans a$,
	which implies $a \in \sem q$ by the definition of sound transition.
	For the inductive case, let $a(t_1, \cdots, t_k) \in \lang q$.
	There exists a sound transition $q \trans a P_1 \cdots P_k$
	s.t.{\@} $t_1 \in \lang {P_1} \cdots t_k \in \lang {P_k}$.
	By inductive hypothesis, $t_1 \in \sem {P_1} \cdots t_k \in \sem {P_k}$.
	Since the transition is sound, $a(t_1, \cdots, t_k) \in \sem q$.
\end{proof}

\begin{lemma}
	\label{lem:apds:sound:init}
	The initial automaton $\ATA A$ is sound.
\end{lemma}

\begin{proof}
	Immediate, since the only internal transitions in $\ATA A$ are those between states in $Q$
	for which $\sem q$ coincides with the language recognized by $q$.
\end{proof}

\begin{lemma}
	\label{lem:apds:sound:induct}
	The saturation process adds only sound transitions.
\end{lemma}

\begin{proof}
	We proceed by case analysis on the newly added rule.
	\paragraph{Case~\ref{rule:apds:popk}.}
	Assume that (1) $p \trans {b(n)} P_1 \cdots P_n P$ is a sound transition.
	We show that (2) $p \trans {a(n)} q^* \cdots q^* (p, a(n), P_{k+1} \cdots P_n q^*) P_{k+1} \cdots P_n q^*$
	and (3) $(p, a(n), P_{k+1} \cdots P_n q^*) \trans {b(k)} P_1 \cdots P_k P$ are sound transitions.
	For (2), assume $t_1 \in \sem {q^*} \cdots t_{k-1} \in \sem {q^*}$,
	$t_k \in \sem{(p, a(n), P_{k+1} \cdots P_n q^*)}$,
	$t_{k+1} \in \sem {P_{k+1}} \cdots t_n \in \sem {P_n}$,
	and $t \in \sem {q^*}$.
	By the definition of $\sem {(p, a(n), P_{k+1} \cdots P_n q^*)}$,
	$a(n)(t_1, \dots, t_n, t) \in \sem p$.
	Notice that we didn't use the fact that (1) is sound, i.e., rules of the form (2) are always sound by construction.
	
	For (3), assume $t_1 \in \sem {P_1} \cdots t_k \in \sem {P_k}$ and $t \in \sem P$,
	and we have to show $b(k)(t_1, \dots, t_k, t) \in \sem {(p, a(n), P_{k+1} \cdots P_n q^*)}$.
	This means that given arbitrary $v_1 \cdots v_{k-1}$ and $v$,
	and $t_{k+1} \in \sem {P_{k+1}} \cdots t_n \in \sem {P_n}$,
	it must be the case that \[ t' := a(n)(v_1, \dots, v_{k-1}, b(k)(t_1, \dots, t_k, t), t_{k+1}, \dots, t_n, v) \in \sem p \ .\]
	Since (1) is sound, we have $b(n)(t_1, \dots, t_n, t) \in \sem p$.
	But \[ t' = a(n)(v_1, \dots, v_{k-1}, b(k)(t_1, \dots, t_k, t), t_{k+1}, \dots, t_n, v) \goesto {} b(n)(t_1, \dots, t_n, t) \in \sem p \]
	in one step by one application of rule~\ref{rule:apds:popk},
	thus $t' \in \sem p$.
	
	\paragraph{Case~\ref{rule:apds:pushk}.}
	Assume that (1) $p \trans {a(n)} P_1 \cdots P_n P$
	and (2) $P_k \trans {a(k)} P_1' \cdots P_k' P'$
	are sound transitions.
	We show that (3) $p \trans {a(n)} P_1 \cup P_1' \cdots P_{k-1} \cup P_{k-1}' P_k' P_{k+1} \cdots P_n P \cup P'$
	is a sound transition.
	Let $t_1 \in \sem {P_1 \cup P_1'} \cdots t_{k-1} \in \sem {P_{k-1} \cup P_{k-1}'}$,
	$t_k \in \sem {P_k'}$,
	$t_{k+1} \in \sem {P_{k+1}} \cdots t_n \in \sem {P_n}$,
	and $t \in \sem {P \cup P'}$,
	and we show $a(n)(t_1, \dots, t_n, t) \in \sem p$.
	This entails that $t_1 \in \sem {P_1} \cap \sem {P_1'} \cdots t_{k-1} \in \sem {P_{k-1}} \cap \sem {P_{k-1}'}$
	and $t \in \sem P \cap \sem {P'}$.
	Since transition (2) is sound,
	$a(k)(t_1, \dots, t_k, t) \in \sem {P_k}$.
	Since transition (1) is sound,
	$a(n)(t_1, \dots, t_{k-1}, a(k)(t_1, \dots, t_k, t), t_{k+1}, \dots t_n, t) \in \sem p$.
	From the operational semantics,
	\[ a(n)(t_1, \dots, t_n, t) \goesto {} a(n)(t_1, \dots, t_{k-1}, a(k)(t_1, \dots, t_k, t), t_{k+1}, \dots t_n, t) \]
	in one step by one application of the \ref{rule:apds:pushk} rule.
	Therefore, $a(n)(t_1, \dots, t_n, t) \in \sem p$.

	\paragraph{Case~\ref{rule:apds:collapse}.}
	Assume that (1) $p \trans {b(n)} P_1 \cdots P_n P$ is a sound rule,
	and we show that (2) $p \trans {a(n)} q^* \cdots q^* (p, a(n))$
	and (3) $(p, a(n)) \trans {b(n)} P_1 \cdots P_n P$ are sound rules.
	For (2), assume $t_1 \dots t_k$ are arbitrary and $t \in \sem {(p, a(n))}$
	By the definition of $\sem {(p, a(n))}$,
	$a(n)(t_1, \dots, t_k, t) \in \sem p$.
	Notice that we used just the definition of $\sem {(p, a(n))}$,
	and not the fact that (1) is sound, or the operational semantics of our system.
	
	For (3), let $t_1 \in \sem {P_1} \dots t_n \in \sem {P_n}$ and $t \in \sem P$,
	and we have to show that $b(n)(t_1, \dots, t_n, t) \in \sem {(p, a(n))}$.
	By the definition of $\sem {(p, a(n))}$,
	this means that for arbitrary $v_1, \dots, v_n$,
	we have to show that
	$a(n)(v_1, \dots, v_n, b(n)(t_1, \dots, t_n, t)) \in \sem p$.
	But rule (1) is sound, therefore $b(n)(t_1, \dots, t_n, t) \in \sem p$,	and
	\[ a(n)(v_1, \dots, v_n, b(n)(t_1, \dots, t_n, t)) \goesto {} b(n)(t_1, \dots, t_n, t) \]
	in one step by an application of rule~\ref{rule:apds:collapse}.
	Thus, $a(n)(v_1, \dots, v_n, b(n)(t_1, \dots, t_n, t)) \in \sem p$.
	
	\paragraph{Case~\ref{rule:apds:pushb}.}
	Assume that (1) $p \trans {b(n)} P_1 \cdots P_n P$,
	(2) $P_1 \trans {a(1)} R_1 R$,
	and	(3) $P \trans {a(n)} R_1' \cdots R_n' R'$ are sound transitions,
	and we show that (4) $p \trans {a(n)} R_1 \cup R_1' P_2 \cup R_2' \cdots P_n \cup R_n' R \cup R'$
	is a sound transition.
	Assume $t_1 \in \sem {R_1 \cup R_1'} = \sem {R_1} \cap \sem {R_1'}$,
	$t_2 \in \sem {P_2 \cup R_2'} = \sem {P_2} \cap \sem {R_2'} \cdots t_n \in \sem {P_n \cup R_n'} = \sem {P_n} \cap \sem {R_n'}$,
	and $t \in \sem {R \cup R'} = \sem R \cap \sem {R'}$,
	and we have to show that $a(n)(t_1, \dots, t_n, t) \in \sem p$.
	Since $t_1 \in \sem {R_1'} \cdots t_n \in \sem {R_n'}$ and $t \in \sem {R'}$
	and (3) is a sound transition, we have $a(n)(t_1, \dots, t_n, t) \in \sem P$.
	Since $t_1 \in \sem R_1$ and $t \in \sem R$,
	and (2) is a sound transition,
	we have $a(1)(t_1, t) \in \sem {P_1}$.
	Finally, since $a(1)(t_1, t) \in \sem {P_1}$,
	$t_2 \in \sem {P_2} \cdots t_n \in \sem {P_n}$
	and $a(n)(t_1, \dots, t_n, t) \in \sem P$,
	and (1) is a sound transition,
	we have $b(n)(a(1)(t_1, t), t_2, \dots, t_n, a(n)(t_1, \dots, t_n, t)) \in \sem p$.
	By the operational semantics,
	\[ a(n)(t_1, \dots, t_n, t) \goesto {} b(n)(a(1)(t_1, t), t_2, \dots, t_n, a(n)(t_1, \dots, t_n, t)) \]
	in one step by the application of rule~\ref{rule:apds:pushb}.
	Therefore, $a(n)(t_1, \dots, t_n, t) \in \sem p$ as required.
	
	\paragraph{Case~\ref{rule:apds:rewb}}
	Assume that (1) $p \trans {b(n)} P_1 \cdots P_n P$ is a sound transition,
	and we show that (2) $p \trans {a(n)} P_1 \cdots P_n P$ is a sound transition.
	Assume $t_1 \in \sem {P_1} \cdots t_n \in \sem {P_n}$ and $t \in \sem P$.
	Since (1) is sound, $b(n)(t_1, \dots, t_n, t) \in \sem p$.
	By the operational semantics,
	\[ a(n)(t_1, \dots, t_n, t) \goesto {} b(n)(t_1, \dots, t_n, t) \]
	by one application of rule~\ref{rule:apds:rewb},
	so $a(n)(t_1, \dots, t_n, t) \in \sem p$ as required.
 	
\end{proof}

\begin{lemma}
	\label{lem:ompds:sound}
	The saturation procedure is sound, i.e., $\lang {\ATA A'} \subseteq \prestar {\lang {\ATA A}}$.
\end{lemma}

\begin{proof}
	By Lemma~\ref{lem:apds:sound:init}, the initial automaton $\ATA A$ is sound,
	and, by Lemma~\ref{lem:apds:sound:induct}, the final automaton $\ATA A'$ is sound, too.
	By Lemma~\ref{lem:apds:sound:incl} applied to the initial state $p$,
	$\lang {\ATA A'} = \lang p \subseteq \sem p = \prestar {\lang {\ATA A}}$.
\end{proof}
}

\subsection{Krivine machine with states}

\label{sec:krivine:machine}

We show that the Krivine machine evaluating simply-typed $\lambda Y$-terms can be encoded as an AOTPS.
Essentially, this encoding was already given in the presentation of the Krivine machine operating on $\lambda Y$-terms from \cite{SalvatiWalukiewicz:Krivine:2011},
though not explicitly given as tree pushdowns.
In this sense, this provides the first saturation algorithm for the Krivine machine,
thus yielding an optimal reachability procedure.
Moreover, in App.~\ref{app:aotps2km} we present also a converse reduction (as announced earlier in Theorem~\ref{thm:aotps2km}),
thus showing that the two models are in fact equivalent.


A \emph{type} is either the basic type 0 or $\alpha\arrow\beta$ for types $\alpha, \beta$.
The \emph{level} of a type is 
$\level 0 = 0$ and $\level {\alpha \arrow \beta} = \max (\level \alpha + 1, \level \beta)$.
We abbreviate $\alpha \arrow \cdots \arrow \alpha \arrow \beta$
as $\alpha^k \arrow \beta$.
Let $\Vars = \set{x_1^{\alpha_1}, x_2^{\alpha_2}, \dots}$ be a countably infinite set of typed variables,
and let $\Gamma$ be a ranked alphabet.
A \emph{term} is either
\begin{inparaenum}[(i)]
	\item a constant $a^{0^k \arrow 0} \in \Gamma$,
	\item a variable $x^{\alpha} \in \Vars$,
	\item an abstraction $(\lambda x^\alpha. M^\beta)^{\alpha \arrow \beta}$,
	\item an application $(M^{\alpha \arrow \beta} N^\alpha)^\beta$, or
	\item a fixpoint $(Y M^{\alpha \arrow \alpha})^\alpha$.
\end{inparaenum}
We sometimes omit the type annotation from the superscript, in order to simplify the notation.
For a given term $M$, its set of \emph{free variables} is defined as usual.
A term $M$ is \emph{closed} if it does not have any free variable. 
We denote by $\Lambda(M)$ be the set of \emph{sub-terms} of $M$.
%
%
%
%
An \emph{environment} $\rho$ 
is a finite type-preserving function assigning closures to variables,
and a \emph{closure} $C^\alpha$ is a pair consisting of a term of type $\alpha$ and an environment,
as expressed by the following mutually recursive grammar:
$\rho ::= \emptyset\ |\ \rho[x^\alpha \mapsto C^\alpha]$ and $C^\alpha ::= (M^\alpha, \rho)$.
%
We say that a closure $(M, \rho)$ is \emph{valid} if $\rho$ binds all variables which are free in $M$ (and no others),
and moreover $\rho(x^\alpha)$ is itself a valid closure for each free variable $x^\alpha$ in $M$.
Sometimes, we need to restrict an environment $\rho$ by discarding some bindings in order to turn a closure $(M, \rho)$ into a valid one.
Given a term $M$ and an environment $\rho$,
the \emph{restriction} of $\rho$ to $M$, denoted $\restrict \rho {M}$,
is obtained by removing from $\rho$ all bindings for variables which are not free in $M$.
In this way, if $(M, \rho)$ is a closure where $\rho$ assigns valid closures to at least all variables which are free in $M$,
then $(M, \restrict \rho {M})$ is a valid closure.
In a closure $(M, \rho)$, $M$ is called the \emph{skeleton}, 
and it determines the type and level of the closure.
Let $\Closures^\alpha(M)$ be the set of valid closures of type $\alpha$ with skeleton in $\Lambda(M)$.
%
%
An \emph{alternating Krivine machine%
\footnote{Cf. also \cite{Krivine:Machine:2007} 
for a definition of the Krivine machine in a different context.}
with states} of level $l \in \Natpos$ is a tuple $\KM = \tuple {l, \Gamma, Q, K^0, \Delta}$,
where 
$\tuple{\Gamma, Q, \Delta}$ is an alternating tree automaton (in which a constant $a^{0^k\arrow 0}\in\Gamma$ is seen as a letter $a$ of rank $k$),
and $K^0$ is a closed term of type $0$ s.t.{\@} the level of any sub-term in $\Lambda(K^0)$ is at most $l$.
%
%
%
In the following, let $\alpha = \alpha_1 \arrow \cdots \arrow \alpha_k \arrow 0$.
The Krivine machine $\KM$ induces a transition system $\tuple{\mathcal C_\KM, \goesto_\KM}$,
where in a configuration $(p, C^\alpha, C_1^{\alpha_1}, \dots, C_k^{\alpha_k}) \in \mathcal C_\KM$,
$p \in Q$, 
$C^{\alpha} \in \Closures^\alpha(K^0)$ is the \emph{head closure},
and $C_1^{\alpha_1} \in \Closures^{\alpha_1}(K^0), \dots, C_k^{\alpha_k} \in \Closures^{\alpha_k}(K^0)$ are the \emph{argument closures}.
%
%
The transition relation $\goesto_\KM$ depends on the structure of the skeleton of the head closure.
It is deterministic 
except when the head is a constant in $\Gamma$,
in which case the transitions in $\Delta$ control how the state changes (cf.{\@} also Fig.~\ref{fig:km:otpds:rules} in the appendix for a pictorial representation):
%
%
\vspace{-.5ex}
	\begin{align*}
		(p, (x^\alpha, \rho), C_1^{\alpha_1}, \dots, C_k^{\alpha_k})\ \goesto_\KM\ 
			& \set {(p, \rho(x^\alpha), C_1^{\alpha_1}, \dots, C_k^{\alpha_k})}, \\
		(p, (M^\alpha N^{\alpha_1}, \rho), C_2^{\alpha_2}, \dots, C_k^{\alpha_k})\ \goesto_\KM\ 
			& \set {(p, (M^\alpha, \restrict \rho {M^\alpha}), (N^{\alpha_1}, \restrict \rho {N^{\alpha_1}}), C_2^{\alpha_2}, \dots, C_k^{\alpha_k})}, \\
		(p, (Y M^{\alpha\arrow\alpha}, \rho), C_1^{\alpha_1}, \dots, C_k^{\alpha_k})\ \goesto_\KM\
			& \set {(p, (M^{\alpha\arrow\alpha}, \rho), ((Y M)^\alpha, \rho), C_1^{\alpha_1}, \dots, C_k^{\alpha_k})}, \\
		(p, (\lambda x^{\alpha_0}. M^\alpha, \rho), C_0^{\alpha_0}, \dots, C_k^{\alpha_k})\ \goesto_\KM\ 
			& \set {(p, (M^\alpha, \rho[x^{\alpha_0} \mapsto C_0^{\alpha_0}]), C_1^{\alpha_1}, \dots, C_k^{\alpha_k})}, \\
		(p, (a^{0^k \arrow 0}, \rho), C_1^0, \dots, C_k^0)\ \goesto_\KM\ 
			&(P_1 \times \set {C_1^0}) \cup \cdots \cup (P_k \times \set {C_k^0}) \\
				&\textrm {for every } p \trans a P_1 \cdots P_k \in \Delta.
	\end{align*}
%
%
We say that $\KM$ is \emph{non-deterministic} if $\tuple{\Gamma, Q, \Delta}$ is non-deterministic and all letters in $\Gamma$ have rank at most $1$.
Given $c\in\mathcal C_\KM$ and $T \subseteq Q$, the \emph{(control-state) reachability problem} for $\KM$ asks whether
$c\in\prestar {T \times (\bigcup_{\alpha = \alpha_1 \arrow \cdots \arrow \alpha_k \arrow 0}
	\Closures^\alpha(K^0) \times \Closures^{\alpha_1}(K^0) \times \cdots \times \Closures^{\alpha_k}(K^0))}$.

\paragraph{Encoding.}

Following \cite{SalvatiWalukiewicz:Krivine:2011},
we encode valid closures and configurations of the Krivine machine as ranked trees.
%
%
%
Fix a Krivine machine $\KM = \tuple {l, \Gamma, Q, K^0, \Delta}$ of level $l$.
We assume a total order on all variables $\tuple {x_1^{\beta_1}, \dots, x_n^{\beta_n}}$ appearing in $K^0$.
For a type $\alpha$, we define $\ord (\alpha) = l - \level \alpha$.
We construct an AOTPS $\Sys = \tuple {l, \Sigma, Q', \Rules}$ of order $l$ as follows.
The ordered alphabet is 
\begin{align*}
	\Sigma = \setof{N^\alpha}{N^\alpha \in \Lambda(K^0)\land\level\alpha<l}\cup\setof {\marked {N^\alpha}} {N^\alpha \in \Lambda(K^0)} \cup \setof{\bot_i}{i\in\set{1,\dots,n}}.
\end{align*}
Here, $N^\alpha$ is a symbol of $\rank (N^\alpha) = n$ and $\ord (N^\alpha) =\ord (\alpha)$.
Moreover, if $\alpha = \alpha_1 \arrow \cdots \arrow \alpha_k \arrow 0$ for some $k \geq 0$,
then $\marked {N^\alpha}$ is a symbol of $\rank (\marked {N^\alpha}) = n + k$
and $\ord (\marked {N^\alpha}) = l$ (in fact, $\ord(\marked{N^\alpha})$ is irrelevant, as $\marked{N^\alpha}$ is used only in the root). 
%
%
Finally, $\bot_i$ is a leaf of order $i$.
The set of control locations is $Q' = Q \cup \bigcup_{(p \trans a P_1 \cdots P_k)\in \Delta}\set{(1, P_1), \dots, (k, P_k)}$.
%
%
A closure $(N^\alpha, \rho)$ is encoded recursively as $\enc {} {N^\alpha, \rho} = N^\alpha(t_1, \dots, t_n)$,
where, for every $i\in\{1,\dots,n\}$,
\begin{inparaenum}[i)]
	\item if $x_i \in \freevars {N^\alpha}$ then $t_i = \enc {} {\rho(x_i)}$, and
	\item $t_i = \bot_{\ord(\beta_i)}$ otherwise (recall that $\beta_i$ is the type of $x_i$).
\end{inparaenum}%
\ignore{
\begin{align*}
	\enc {} {N^\alpha, \rho} = \begin{tikzpicture}[baseline=-2.3ex,scale=0.5,level/.style={sibling distance = 3cm/#1}]
						\node {$N^\alpha$}
						    child{ node {$\enc {} {\rho(x_1)}$} }
						    child{ node {$\cdots$} edge from parent[draw=none] }
						    child{ node {$\enc {} {\rho(x_n)}$} }
						; 
					\end{tikzpicture}, 
\end{align*}
}%
%
%
A configuration $c = (p, (N^\alpha, \rho), C_1^{\alpha_1}, \dots, C_k^{\alpha_k})$
is encoded as the tree $\enc {} c = \marked {N^\alpha}(t_1, \dots, t_n, \enc {} {C_1^{\alpha_1}}, \dots, \enc {} {C_k^{\alpha_k}})$,
where the first $n$ subtrees encode the closure $(N^\alpha, \rho)$,
i.e., $\enc {} {N^\alpha, \rho} = N^\alpha(t_1, \dots, t_n)$.
%
\ignore{
\begin{align*}
	\enc {} {p, (N^\alpha, \rho), C_1^{\alpha_1}, \dots, C_k^{\alpha_k}} =
		\!\!\!\!\!\!\!\!\!\!\!\!\!\!\!\!\!\!\!\!\!\!\!\!\!\!\!\!\!\!\!\!\!\!\!\!\!\!\!\!\!\!\!\!
		\begin{tikzpicture}[baseline=-1ex,scale=0.5,level/.style={sibling distance = 3.3cm/#1}]
			\node {$\marked {N^\alpha}$}
			    child{ node {$\enc {} {\rho(x_1)}$} }
			    child{ node {$\cdots$} edge from parent[draw=none] }
			    child{ node {$\enc {} {\rho(x_n)}$} }
			    child{ node {$\enc {} {C_1^{\alpha_1}}$} }
			    child{ node {$\cdots$} edge from parent[draw=none] }
			    child{ node {$\enc {} {C_k^{\alpha_k}}$} }
			; 
		\end{tikzpicture}
\end{align*}
}%
The encoding is extended point-wise to sets of configurations.
Notice that $K^0$ uses only variables of level at most $l-1$ (the subterm $\lambda x^\alpha.N$ introducing $x^\alpha$ is of level higher by one), so all skeletons in an environment are of order at most $l-1$.
Similarly, skeletons in argument closures are of level at most $l-1$; only the head closure may have a skeleton of level $l$.
Thus we do not need symbols $N^\alpha$ for $\level\alpha=l$.

%
Below, we assume that $\alpha = \alpha_1 \arrow \cdots \arrow \alpha_k
\arrow 0$, that variable $\var y_j$ has order
$\ord (\alpha_j)$ for every $j\in\{0,\dots,k\}$, and that  
variables $\var x_i$ and $\var z_i$ have order $\ord (\beta_i)$ for
every $i\in\{1,\dots,n\}$. 
Notice that $\ord (\alpha) < \ord (\alpha_1), \dots, \ord (\alpha_k)$.
Moreover, we write $\vec {\var x} = \tuple{\var x_1, \dots, \var x_n}$,
$\vec {\var z} = \tuple{\var z_1, \dots, \var z_n}$,
and $\vec {\var y} = \tuple {\var y_1, \dots, \var y_k}$.
Finally, by $\restrict {\vec {\var x}} {M}$ we mean the tuple which is the same as $\vec {\var x}$,
except that positions corresponding to variables not free in $M$ are replaced by the symbol $\bot_{\ord(\beta_i)}$.
%
%
$\Rules$ contains the following rules:
\begin{align*}
	&p, \marked {x_i^\alpha}(\var z_1, \dots, \var z_{i-1}, M^\alpha(\vec {\var x}), \var z_{i+1}, \dots, \var z_n, \vec {\var y})
		\goesto 
		\set p, \marked {M^\alpha}(\vec {\var x}, \vec {\var y}),
	\\
 	&p, \marked {M^\alpha N^{\alpha_1}}(\vec {\var x}, \var y_2, \dots, \var y_k)
		\goesto
		\set p, \marked {M^\alpha}(\restrict {\vec {\var x}} {M^\alpha}, N^{\alpha_1}(\restrict {\vec {\var x}} {N^{\alpha_1}}), \var y_2, \dots, \var y_k),
	\\
	&p, \marked {Y M^{\alpha \arrow \alpha}}(\vec {\var x}, \vec {\var y})
		\goesto
		\set p, \marked {M^{\alpha \arrow \alpha}}(\vec {\var x}, Y M^{\alpha \arrow \alpha}(\vec {\var x}), \vec {\var y}),
	\\
	&p, \marked {\lambda x_i^{\alpha_0}. M^\alpha}(\vec {\var x}, \var y_0, \vec {\var y})
		\goesto
		\set p, \marked {M^\alpha}(\var x_1, \dots, \var x_{i-1}, y_0, \var x_{i+1}, \dots, \var x_n, \vec {\var y}),\\
	&p, \marked {a^{0^k \arrow 0}}(\vec {\var x}, \vec {\var y})
		\goesto
		\set{(1, P_1), \dots, (k, P_k)}, \marked {a^{0^k \arrow 0}}(\vec {\var x}, \vec {\var y})
		\qquad\qquad
		\forall (p \trans a P_1 \cdots P_k)\in \Delta,
	\\
	&(i, P_i), \marked {a^{0^k \arrow 0}}(\vec {\var z}, \var y_1, \dots, \var y_{i-1}, M_i^0(\vec {\var x}), \var y_{i+1}, \dots, \var y_k)
		\goesto
		P_i, \marked {M_i^0}(\vec {\var x}).
\end{align*}
The first rule satisfies the ordering condition
since the shared variables $\var y_i$ are of order strictly higher than $\ord(M^\alpha)$.
\ignore{
\begin{remark}
	Notice that, except for the abstraction rule \ref{rule:abs2},
	in every other rule the left-hand side variables also appear in the right-hand side.
	Moreover, all rules except the application \ref{rule:app} and fixpoint \ref{rule:fix} rule are linear.
	Finally, if the original lambda-expression is linear (as a lambda expression),
	then the application rule becomes linear.
\end{remark}
}%
%
%
%
	%
%
%
A direct inspection of the rules shows that, for a configuration $c$ and a set of configurations $D$,
we have $c \goesto_\KM^* D$ if, and only if,
$\enc {} {c} \goesto_\Sys^* \enc {} {D}$.
Therefore, the encoding preserves reachability properties.
Since a Krivine machine of level $n$ is simulated by a flat AOTPS of order $n$,
the following is an immediate consequence of Theorems~\ref{thm:complexity} and \ref{thm:complexity:nondeterministic}.
\begin{theorem}[\cite{aehlig07}]
	\label{thm:krivine:complexity}
	Reachability in alternating Krivine machines with states of level $n$
	and in non-deterministic Krivine machines with states of level $n+1$ is \nexptimec n.
\end{theorem}



\subsection{Ordered annotated multi-pushdown systems}

\label{sec:oapds}

Ordered annotated multi-pushdown systems are the common generalization of ordered multi-pushdown systems and annotated pushdown systems \cite{Hague:FSTTCS:2013}.
Such a system is comprised of $m > 0$ annotated higher-order pushdowns arranged from left to right, where each pushdown is of order $n > 0$.
While push operations are unrestricted, pop and collapse operations implicitly destroy all pushdowns to the left of the pushdown being manipulated,
in the spirit of \cite{BreveglieriCherubiniCitriniCrespi-Reghizzi:Ordered:1996,AtigBolligHabermehl:Ordered:2008,Atig:ordered:2012}.
\cite{Hague:FSTTCS:2013} has shown that reachability in this model can be decided in $mn$-fold exponential time,
by using a saturation-based construction leveraging on the previous analysis for the first-order case \cite{BreveglieriCherubiniCitriniCrespi-Reghizzi:Ordered:1996,AtigBolligHabermehl:Ordered:2008,Atig:ordered:2012}.
In App.~\ref{sec:oapds:app}, we provide a simple encoding of an annotated multi-pushdown system with parameters $(m, n)$
into an AOTPS of order $m n$.
It is essentially obtained by taking together our previous encodings of ordered (cf.{\@} Sec.~\ref{sec:ordered:mpds}) and annotated systems (cf.{\@} Sec.~\ref{sec:annotated:pda}).
As a consequence of this encoding,
by using the fact that an AOTPS of order $m n$ can be encoded by a Krivine machine of the same level (by Theorem.~\ref{thm:aotps2km}),
and by recalling the known fact that the latter can be encoded by a 1-stack annotated multi-pushdown system of order $m n$ \cite{Salvati:Walukiewicz:MSCS:2015},
we deduce that the concurrent behavior of an ordered $m$-stack annotated multi-pushdown system of order $n$
can be \emph{sequentialized} into a 1-stack annotated pushdown system of order $m n$ (thus at the expense of an increase in order).
The following complexity result is a direct consequence of Theorem~\ref{thm:complexity}.

\begin{theorem}[\cite{Hague:FSTTCS:2013}]
	Reachability in alternating ordered annotated multi-pushdown systems of parameters $(m, n)$ is in \nexptime {(mn)}.
\end{theorem}

We remark that our result is for alternating systems, while in \cite{Hague:FSTTCS:2013} they consider non-deterministic systems and obtain \nexptime{(m(n-1))} complexity.
It seems that their method can be extended to alternating systems, and then the complexity becomes \nexptime {(mn)} as well.



\section{Safety}
\label{sec:safety}

\ignore{

When the order is monotone w.r.t. the nesting of subtrees, we say that the tree is safe.
Formally, a tree $u$ is \emph{safe} if every subtree $t$ thereof has order $\ord(t) \leq \ord(u)$ and it is itself safe.
Safe trees will arise as a structural invariant in some of the translations presented in Sec.~\ref{sec:examples}.

A rewrite rule $l \rewritesto r$ is \emph{safe} if both $l$ and $r$ are safe.

We say that $\Sys$ is \emph{safe} if all its rules are safe.

We observe that, if we apply the encoding above to higher-order pushdown stack (i.e., without stack annotations),
then we obtain \emph{safe} trees.
Indeed, the only source of non-safety in the encoding tree $\tuple {a, k}(\enc 1 {u_1}, \dots, \enc k {u_k}, \enc i u)$
is in the annotation $u$ which can have rank $i$ higher than $k$ in general.
A similar translation from higher-order pushdown systems (i.e., without annotations) can be given,
yielding a \emph{safe} AOTPS by the considerations above.

\noindent
We remark that we can obtain a \emph{safe} AOPDS from the translation above
by imposing a simple structural condition on lambda terms (also called safety);
cf.{\@} App.~\ref{app:krivine:safety} for the details.

}

The notion of safety has been made explicit by Knapik,
Niwi\'nski, and Urzyczyn~\cite{KnapikNiwinskiUrzyczyn:Easy:2002} who identified the class
of \emph{safe recursive schemes}. They have shown that this class
defines the same set of infinite trees as higher-order pushdown
systems, i.e., the systems from Sec.~\ref{sec:annotated:pda} but
without annotations. Blum and Ong~\cite{Blum:Ong:2009} have extended 
the notion of safety to the simply-typed
$\lambda$-calculus in a clear way.  Then~\cite{Salvati:Walukiewicz:MSCS:2015} adapted it to
$\lambda Y$-calculus, and have shown that safe $\lambda Y$-terms correspond to
higher-order pushdown automata without annotation.

There is a simple notion of safety for AOTPSs that actually corresponds
to safety  for pushdown systems and terms.  We say
that a $(\Sigma\cup \Vars)$-tree is \emph{safe} when looking from the
root to the leafs the order does never increase.
Formally, a tree $u$ is \emph{safe} if every subtree $t$ thereof has
order $\ord(t) \leq \ord(u)$ and it is itself safe.
A rewrite rule $l \rewritesto r$ is \emph{safe} if both $l$ and $r$ are safe.
We say that $\Sys$ is \emph{safe} if all its rules are safe.

As a first example, let us look at the encoding of annotated
higher-order pushdown systems from Sec.~\ref{sec:annotated:pda}. If
we drop annotation then higher-order pushdowns are represented by safe
trees, and all the rules are safe in the sense above.
The case of Krivine machines is more difficult to explain,
because it would need the definition of safety from~\cite{Salvati:Walukiewicz:MSCS:2015}.
In particular, one would have to partition variables into \emph{lambda-variables} and \emph{$Y$-variables},
which we avoid in the current presentation for simplicity.
In the full version of the paper we will show that safe terms are encoded by safe trees,
and that all the rules of the encoding of the Krivine machine preserve safety.
Finally, we remark that the translation from AOTPSs to the Krivine machine with states previously announced in Theorem~\ref{thm:aotps2km}
can be adapted to produce a safe Krivine machine with states from a safe AOTPS.


\ignore{


 In order to examine the case of Krivine machines we formulate a
 definition of safe term. A term is $M$ is called \emph{locally unsafe}
 if it has free variables of level strictly smaller than $\level M$,
 otherwise it is \emph{locally safe}. An occurrence of a term $N$ as a
 subterm of a term $M$ is \emph{safe} if it is in the context
 $\dots (N\,K)\dots$.  A term is \emph{safe}, if all occurrences of
 locally unsafe subterms are safe.  An alternating Krivine Machine with
 states $\calM=\langle l,\Gamma,Q,M^0,\Delta\rangle$ is \emph{safe} if
 $M^0$ is safe. A configuration
 $(p,C^\tau,C_1^{\tau_1},\dots,C_k^{\tau_k})$ is \emph{safe} if in all
 its closures (also recursively inside environments), except
 potentially the head closure, skeletons are safe.  Notice that when
 the machine is safe, the skeleton of the head closure has to be a
 subterm of a safe term, but it need not to be safe itself (all
 occurrences of locally unsafe proper subterms are safe, but the whole
 term may be locally unsafe).  It can be checked that all
 configurations of the computation of the the safe Krivine machine are
 safe.

 Let us now look at our encoding of Krivine machines with states into
 AOTPSs. It follows that a safe configuration is encoded in a safe
 tree.  Indeed, $\enc{}{N^\tau,\rho}$ starts with a node of order
 $\ord(\tau)$, and since $N$ is safe (and in particularly locally safe)
 then in $\rho$ we have only closures of levels not smaller than
 $\mathsf{level}(\tau)$; the encodings of these closures will have
 orders not greater than $\ord(\tau)$.  The head closure (whose
 skeleton need not to be locally safe) is not encoded this way, but
 rather as the root labeled by a $[N]$ letter of maximal order $l$.
 It can be checked that the rules of the created AOTPS are safe as well.
%
%
}

\section{Conclusions}

\label{sec:conclusions}

We have introduced a novel extension of pushdown automata which is
able to capture several sophisticated models thanks to a simple
ordering condition on the tree-pushdown.  While ordered tree-pushdown
systems are not more expressive than annotated higher-order pushdown
systems, or than Krivine machines, they offer some conceptual
advantages. 
Compared to Krivine machines, they have states, and
typing is replaced by a lighter mechanism of ordering; for example, the
translation from our model back to the Krivine machine is much more cumbersome. Compared to
annotated pushdown automata, the tree-pushdown is more versatile than
a higher-order stack; for example,
one can compare the encoding of the Krivine machine into our model
to its encoding to annotated pushdown automata. We hope that ordered tree-pushdown
systems will help to establish more connections with other models, as
we have done in this paper with multi-pushdown systems.

There exist restrictions of multi-pushdown systems that we do not cover in this paper.
Reachability games are decidable for
phase-bounded multi-pushdown systems~\cite{Seth:Global:2010}. We can encode
the phase-bounded restriction directly in our tree-pushdown systems, but
we do not know how to deal with the scope-bounded restriction. Encoding
the scope-bounded restriction would give an algorithm for reachability
games over such systems, but we do not know if the problem is decidable.

%
Our general saturation algorithm can be used to verify reachability properties.
We plan to extend it to the more general \emph{parity properties},
in the spirit of \cite{HagueOng:SaturationPDS:IC2011}. 
We leave as future work implementing our saturation algorithm,
leveraging on subsumption techniques to keep the search space as small as possible.


\paragraph{\bf Acknowledgments.}
We kindly acknowledge stimulating discussions with Ir\`ene Durand, G\'eraud S\'enizergues, and Jean-Marc Talbot,
and the anonymous reviewers for their helpful comments.

\bibliographystyle{abbrv.bst}
\bibliography{local} 
\clearpage
\appendix

\section{Proof of Lemma~\ref{lem:saturation:correct}}

\label{app:correctness}

Let $\ATA A$ be the automaton recognizing the target set of configurations,
and let $\ATA B$ be the automaton obtained at the end of the
saturation procedure (cf.~page~\pageref{def:B}).

\lemsaturationiscorrect*

\noindent
We prove the two inclusions of the lemma separately.

\begin{lemma}[Completeness]
	For $\ATA A$ and $\ATA B$ as above,
	$\prestar {\lang {\ATA A, P}} \subseteq \lang {\ATA B, P}$.
\end{lemma}

\begin{proof}
	Let $(p, t)$ be a configuration in $\prestar {\lang {\ATA A, P}}$. 
	We show $(p, t) \in \lang {\ATA B, P}$ by induction on the length $d \geq 0$
	of the shortest sequence of rewrite steps from $(p, t)$ to $\lang {\ATA A, P}$.
	If $d = 0$, then $(p, t) \in \lang {\ATA A, P}$.
	Since the saturation procedure only adds states and transitions to $\ATA A$,
	we directly have $(p, t) \in \lang{\ATA B, P}$.
	Inductively, assume that the property holds for all configurations reaching $\lang {\ATA A, P}$ in at most $d \geq 0$ steps,
	and let configuration $(p, t)$ be at distance $d+1 > 0$ from $\lang {\ATA A, P}$.
	There exists a rule $(p, l \rewritesto S, r) \in \Rules$ 
	and a substitution $\sigma$ s.t.{\@} $t = l\sigma$ and from every configuration in $S\times\{r\sigma\}$ we can reach $\lang {\ATA A, P}$ in at most $d$ steps.
	%
	%
	Let $l=a(u_1,\dots,u_m)$ and $t=a(t_1, \dots, t_m)$.
	By induction hypothesis, $S \times \set{r\sigma} \subseteq \lang{\ATA B, P}$, thus $\ATA B$ has a run tree $\beta$ from $S$ on $r\sigma$.
	Its part, also denoted $\beta$, is a run tree from $S$ on $r$.
	Suppose first that our rule is shallow, and
	consider the transition $p \trans a P_1 \cdots P_m$ added to $\Delta'$ in the saturation procedure because of this rule $p,l\rewritesto S,r$ and this run tree $\beta$.
	This transition can be used in the root of $t$, so it suffices to show that $t_1 \in \lang {P_1}, \dots, t_m \in \lang {P_m}$.
	If $u_i$ is $r$-ground, then $P_i = \set{p^{u_i}}$ and $t_i \in \lang {p^{u_i}}$ by construction.
	If $u_i = \var x$ is a variable appearing in $r$, then by definition $P_i = \bigcup \beta(r^{-1}(\var x))$.
	Since $\beta$ is a run tree on the whole $r\sigma$, we have $t_i=\sigma(\var x)\in\lang{P_i}$.
	The case of a deep rule is similar. 
	Let $u_k=b(v_1, \dots, v_{m'})$ be the lookahead subtree of $l$ that is neither $r$-ground nor a variable, and let $t_k = b(s_1, \dots, s_{m'})$.
	Then in the root of $t_k$ we use the second added transition $(g, P_{k+1}, \dots, P_m) \trans b S_1 \cdots S_{m'}$.
	It suffices to show $s_1 \in \lang {S_1}, \dots, s_{m'} \in \lang {S_{m'}}$, which is done as above.
\end{proof}

\begin{lemma}[Soundness]
	\label{lem:soundness}
	For $\ATA A$ and $\ATA B$ as above,
	$\lang {\ATA B, P} \subseteq \prestar {\lang {\ATA A, P}}$.
\end{lemma}

\noindent
The soundness proof requires several steps.
First, we assign a semantics $\sem p \subseteq \Trees{\Sigma}$ to all states $p$ in $\ATA B$.
For a set of states $S \subseteq Q'$, $\sem S := \bigcap_{p\in S} \sem p$.
For $p\in Q'_{n+1}$ we take
\begin{align*}
	\sem p := \left\{ \begin{array}{ll}
		\setof {t} {(p, t) \in \prestar {\lang {\ATA A, P}}} & \textrm{if } p \in P, \\
		\mathcal L_{\ATA A}(p) & \textrm{if } p \in Q \setminus P,\\
		\mathcal L_{\ATA B}(p) & \textrm{if } p=p^v\in Q_0.
	\end{array}\right.
\end{align*}
Then by induction on $n-i$ we define $\sem p$ for $p=((q,l\rewritesto S,r),P_{k+1},\dots,P_m)\in Q'_i\setminus Q'_{i+1}$.
Let $l=a(u_1,\dots,u_m)$, where $u_k$ is the lookahead subtree.
As $\sem p$ we take the set of trees $t_k\in\sem{p^{u_k}}$ s.t.{\@} for all $t_1\in\sem{p^{u_1}},\dots,t_{k-1}\in\sem{p^{u_{k-1}}}$ and all $t_{k+1} \in \sem {P_{k+1}}, \dots, t_m \in \sem {P_m}$
it holds $(q, a(t_1, \dots, t_m)) \in \prestar {\lang {\ATA A, P}}$.
Notice that $P_{k+1}, \dots, P_m \subseteq Q_{i+1}'$, so $\sem {P_{k+1}}, \dots, \sem {P_m}$ as well as $\sem{p^{u_1}},\dots,\sem{p^{u_k}}$ are already defined.
%
%
Second, we define sound transitions as those respecting the semantics.  
Formally, a transition $p \trans a P_1 \cdots P_m$ is \emph{sound} iff
$\forall (t_1 \in \sem {P_1}, \dots, t_m \in \sem {P_m})$, $a(t_1, \dots, t_m) \in \sem p$.
%
%
\begin{proposition}\label{prop:sound}
	If all transitions are sound, then $\lang p \subseteq \sem p$ for every $p \in Q'$.
\end{proposition}
\begin{proof}
	Let $t \in \lang p$.
	We proceed by complete induction on the height of $t$.
	%
	Let $t = a(t_1, \dots, t_m)$ (possibly $m=0$ if $t=a$ is a leaf).
	There exists a sound transition $p \trans a P_1 \cdots P_m$
	s.t.{\@} $t_1 \in \lang {P_1}, \dots, t_m \in \lang {P_m}$.
	By induction hypothesis, $t_1 \in \sem {P_1}, \dots, t_m \in \sem {P_m}$,
	and thus by the definition of sound transition, $a(t_1, \dots, t_m) \in \sem p$.
\end{proof}

\begin{proposition}\label{prop:sound:initial}
	Transitions in $\Delta \cup \Delta_0$ are sound.
\end{proposition}
\begin{proof}
	Let $(p \trans a P_1 \cdots P_m) \in \Delta$,
	and let $t_1 \in \sem {P_1}, \dots, t_m \in \sem {P_m}$.
	Since we assume that there are no transitions back to the initial states in $P$,
	we have $P_1, \dots, P_m \subseteq Q \setminus P$,
	and thus $t_1 \in \mathcal L_{\ATA A}(P_1), \dots, t_m \in \mathcal L_{\ATA A}(P_m)$ by the definition of the semantics.
	Consequently, $t := a(t_1, \dots, t_m) \in \mathcal L_{\ATA A}(p)$.
	If $p \not\in P$ we are done, since $\sem p = \mathcal L_{\ATA A}(p)$ in this case.
	Otherwise, if $p \in P$ then $(p, t) \in \lang {\ATA A, P}$, which is included in $\prestar {\lang{\ATA A, P}}$,
	and thus we have $t \in \sem p$ by definition.
	
	For $\Delta_0$ the situation is even simpler.
	Let $p \trans a P_1 \cdots P_m \in \Delta_0$,
	and let $t_1 \in \sem {P_1}, \dots, t_m \in \sem {P_m}$.
	By the definition of the semantics we have $t_1 \in \mathcal L_{\ATA B}(P_1), \dots, t_m \in \mathcal L_{\ATA B}(P_m)$ which implies that $t := a(t_1, \dots, t_m) \in \mathcal L_{\ATA B}(p)=\sem p$.
\end{proof}


%
%
%
%
%
%
%
%
\begin{proposition}\label{prop:sound:inductive}
	The saturation procedure adds only sound transitions.
\end{proposition}
\begin{proof}
	This is induction on the order in which transitions are added by the procedure.
	Let $g = (p, l \rewritesto S, r)$ with $l = a(u_1, \dots, u_m)$,
	and let $t$ be a run tree in $\ATA B$ from $S$ on $r$.
	Since all transitions used in $t$ were present in $\ATA B$ earlier, they are sound.
	We show that the transition $p \trans a P_1 \cdots P_m$ as added by saturation is sound.
	To this end, let $t_1 \in \sem {P_1}, \dots, t_m \in \sem {P_m}$,
	and we show $t' := a(t_1, \dots, t_m) \in \sem p$.
	Since $p \in P$, this amounts to showing that $(p, t') \in \prestar {\lang {\ATA A, P}}$.

	First, assume that $g$ is shallow.
	Observe that $t'=l\sigma$ for some substitution $\sigma$ (if $u_i$ is $r$-ground then $P_i=\{p^{u_i}\}$, so $t_i\in\sem{P_i}$ means that $t_i$ ``matches'' to $u_i$;
	if $u_i=\var x$ is a variable appearing in $r$ then $P_i$ is nonempty and contains states of order $\ord(\var x)$, so $t_i$ is of the same order as $\var x$).
	Thus the system has a transition from $(p,t')$ to $S \times \set{r\sigma}$, 
	and it thus suffices to show $S \times \set{r\sigma} \subseteq \prestar {\lang {\ATA A, P}}$.
	Every node of $r$ labeled by a variable $\var x$ is labelled in the run tree $t$
	by a subset of $P_i = \bigcup t(r^{-1}(\var x))$ for some $i$, and simultaneously $\sigma(\var x)=t_i$ (recall that all variables of $r$ have to appear in $l$).
	Since $t$ uses only sound transitions and $t_1 \in \sem {P_1}, \dots, t_m \in \sem {P_m}$,
	by induction on its height we have $r\sigma \in \sem S$,
	which implies $S \times \set{r\sigma} \subseteq \prestar {\lang {\ATA A, P}}$
	by the definition of the semantics since $S \subseteq P$.
	
	If $g$ is deep, then $P_k = \set{(g, P_{k+1}, \dots, P_m)}$.
	Recall our assumption that $u_1,\dots,u_{k-1}$ have order at most $\ord(u_k)$;
	due to the ordering condition they are $r$-ground.
	It follows that $P_1=\{p^{u_1}\},\dots,P_{k-1}=\{p^{u_{k-1}}\}$, so $t_1\in\sem{p^{u_1}},\dots,t_{k-1}\in\sem{p^{u_{k-1}}}$.
	Since $t_k \in \sem{P_k},\dots,t_m\in\sem{P_m}$, 
	we deduce directly from the definition of $\sem{P_k}$ that $(p, t') \in \prestar {\lang {\ATA A, P}}$.
	
	When $g$ is deep,
	the transition $(g, P_{k+1}, \dots, P_m) \trans b S_1 \cdots S_{m'}$
	is additionally added for this rule,
	and we have to show that this transition is sound too.
	Let $w_1 \in \sem {S_1}, \dots, w_{m'} \in \sem {S_{m'}}$,
	and we show $t_k := b(w_1, \dots, w_{m'}) \in \sem {(g, P_{k+1}, \dots, P_m)}$.
	To this end, let $t_1\in\sem{p^{u_1}},\dots,t_{k-1}\in\sem{p^{u_{k-1}}}$ and $t_{k+1} \in \sem {P_{k+1}}, \dots, t_m \in \sem {P_m}$,
	and we show $(p, a(t_1, \dots, t_m)) \in \prestar {\lang {\ATA A, P}}$.
	The proof is as for a shallow rule, noticing that a node labeled in $r$ by a variable $\var x$ is labeled in $t$ 
	either by a subset of $P_i$ for some $i\in\{k+1,\dots,m\}$ (and then $\sigma(\var x)=t_i$), or by a subset of $S_j$ for some $j\in\{1,\dots,m'\}$ (and then $\sigma(\var x)=w_j$);
	we can again conclude that $r \sigma \in \sem S$ by induction on the height of $t$.
\end{proof}

\begin{proof}[Proof of Lemma~\ref{lem:soundness}]
	By \propref{prop:sound:initial}, the initial transitions in $\Delta \cup \Delta_0$ are sound,
	and by \propref{prop:sound:inductive}, all transitions in $\Delta'$ are sound.
	Let $(p, t) \in \lang {\ATA B, P}$. Thus, $t \in \lang p$.
	By \propref{prop:sound}, $t \in \sem p$.
	Since $p \in P$, by the definition of the semantics,
	$(p, t) \in \prestar {\lang {\ATA A, P}}$.
\end{proof}


\section{Complexity of control-state reachability for flat non-deterministic AOTPSs}
\label{app:complexity:nondeterministic}

While control-state reachability (reachability of a configuration having a particular control location)
for order-$n$ annotated pushdown systems is \nexptimec n (cf.{\@} Theorem~\ref{thm:apds:complexity})%
---and similarly for the Krivine machine (cf.{\@} Theorem~\ref{thm:krivine:complexity})---%
it is known that if we consider \emph{non-deterministic} annotated pushdown systems of order $n$,
the complexity goes down to \nexptimec {(n-1)} (similarly for an analogous restriction on the Krivine machine).
As we will show below, this is also the case for \emph{flat} AOTPSs.

\thmcomplexitynondeterministic*

\noindent
In fact, we prove a stronger statement (cf.{\@} Theorem~\ref{thm:non-alt} below).
Instead of control-state reachability,
we consider reachability of target sets defined by a restricted class of alternating tree automata,
which we call \emph{non-$n$-alternating alternating tree automata}.
Intuitively, this class of tree automata will be defined in such a way that
\begin{itemize}
	\item it is preserved by the saturation procedure, and
	\item allows a faster running time of the procedure by saving one exponential in the number of states (and, consequently, transitions).
\end{itemize}

\noindent
Formally, an alternating tree automaton $\ATA A = \tuple {\Sigma, Q, \Delta}$ is \emph{non-$n$-alternating}%
\footnote{A similar notion for collapsible pushdown systems was proposed in \cite{BroadbentCarayolHagueSerre:Saturation:2012}.}
if its state space $Q$ can be partitioned into two sets $Q_-$ and $Q_+$ so that:
\begin{itemize}
\item	for every transition $p\xrightarrow{a} P_1\dots P_m \in \Delta$ with $p\in Q_-$ it holds $P_i\subseteq Q_-$ for all $i$,
\item	for every transition $p\xrightarrow{a} P_1\dots P_m \in \Delta$ with $p\in Q_+$ it holds $\sum_i|P_i\cap Q_+|\leq 1$, and
\item	in $Q_-$ there is exactly one state of order $n$, call it $p^n$, and the set of transitions from $p^n$ is $\{p\xrightarrow{a}\emptyset\dots\emptyset\mid\ord(a)=n\}$.
\end{itemize}
That is, states in $Q_-$ are closed under the transition relation.
Moreover, they are either of order $\leq n-1$, or trivially accept every order-$n$ tree
(i.e, $p^n$ is the unique such state in $Q_-$).
On the other side, a state $p$ in $Q_+$ can be of order $n$,
but then it can non-trivially accept at most one subtree of order $n$ (by staying in $Q_+$).
When this happens, there is no alternation when doing so, i.e.,
this unique subtree of order $n$ is accepted by a single state in $Q_+$.

\begin{theorem}\label{thm:non-alt}
	Reachability of a target set defined by a non-$n$-alternating alternating tree automaton
	in order-$n$ non-deterministic flat AOTPSs is \nexptimec {(n-1)}, where $n\geq 2$.
\end{theorem}

\noindent
In particular we can realize control-state reachability,
since one can build a non-$n$-alternating automaton that only checks the control-state and accepts every tree.
%
%
The better complexity bound described by Theorem \ref{thm:non-alt} is realized by almost the same saturation procedure as in the general case; we only perform two small modifications.
First, for different variables $\var x$ of order $n$,
we have separate states $p^\var x$.
All these states do the same: they just check that the order of the node is $n$.
Moreover, we already have such state in $\calA$: it is called $p^n$ in the definition of a non-$n$-alternating automaton.
This redundancy should be eliminated: instead of using all these states we glue them together into this single state $p^n$.
Second, we have created states $p^v$ for every subtree $v$ of the l.h.s.{\@} $l$ of every rule.
But we need such states only for $v$ other than the whole $l$, and, for a deep rule, than the lookahead subtree of $l$.
For convenience, let us keep only such $p^v$ and remove all other.
It is easy to see that these modifications do not influence correctness.
Thus, we only need to analyze the complexity.

First, we observe that the automaton obtained at any step of the saturation procedure is non-$n$-alternating,
and we will later analyze its size depending on this assumption.
We divide the states of $\ATA B$ in $Q'$ into $Q'_-$ and $Q'_+$ as follows
(as required in the definition of a non-$n$-alternating automaton).
The original automaton $\ATA A$ by assumption is non-$n$-alternating,
which by definition gives us a partitioning of states from $Q$ into $Q_+$ and $Q_-$.
Recalling that we can assume that in $\ATA A$ we do not have transitions leading to initial states
(we already made and justified this same assumption when describing the saturation procedure),
we assume that all initial states are in $Q_+$.
We inherit this partitioning for states in $Q'$.
The states from $Q_0$ are all taken to $Q'_-$.
%
Because the system is flat, none of states $p^v$ is of order $n$ (we have such states only for $v$ being a variable of order $\leq n-1$), and thus
$\langle\Sigma,Q\cup Q_0,\Delta\cup\Delta_0\rangle$ is non-$n$-alternating for this division of states.
Next, by induction on $n-i$ we classify states from $Q_i'$.
Consider a state $(g,P_{k+1},\dots,P_m)\in Q_i'\setminus Q_{i+1}'$.
We put it into $Q'_+$ if $P_{k+1}\cup\dots\cup P_m\subseteq Q'_-$;
otherwise (some state from some $P_i$ is in $Q'_+$), we put it into $Q'_-$.
In particular, for $i=n$ the state is always taken to $Q'_+$, as necessarily $k=m$.
Recall that the states in the sets $P_j$ come from $Q_{i+1}'$,
so for them we already know whether they are in $Q'_+$ or in $Q'_-$.
We have not added any transitions, so $\langle\Sigma,Q_1',\Delta\cup\Delta_0\rangle$ is still non-$n$-alternating for the above division of states.

Now we should see that a single step of the saturation procedure preserves the property that the automaton is non-$n$-alternating for our division of states.
We only need to check that the newly added transitions satisfy the required properties.
By induction assumption we know that the automaton before the considered step is non-$n$-alternating.
Moreover the system is non-deterministic, so $|S|=1$ in the considered rule $p,l\to S,r$.
This ensures the following property of the run tree $t$ from $S$ on $r$:
State sets on only one path may contain states from $Q'_+$, each set at most one such state.
In particular at most one leaf is labelled by such state set.
Thus, at most one among the sets $\bigcup t(r^{-1}(\var x))$ contains a state from $Q'_+$,
and it contains at most one such state; the sets $\{p^\var x\}$ do not contain states from $Q'_+$.
For a shallow rule this is the end, as all $P^t_j$'s are of this form (with different variables for different $j$'s).
For a deep rule we also have the special child with $P_b^t=\set{(g, P^t_{k+1},\dots,P^t_m)}$.
When none of the $P_j^t$'s for $j\neq k$ has a state from $Q'_+$, then by definition $(g, P^t_{k+1},\dots,P^t_m) \in Q'_+$;
the transition from $p$ is fine (recall that the initial state $p$ belongs to $Q'_+$),
and the transition from $(g, P^t_{k+1},\dots,P^t_m)$ is also fine since at most one among $S_j^t$'s has a state from $Q'_+$.
Suppose that some $P_j^t$ for $i\neq k$ has a state from $Q'_+$.
This is only possible when the corresponding subtree of $l$ is not $r$-ground,
so it is of order greater than the special deep subtree $b(v_1, \dots, v_{m'})$ of $l$ thanks to the ordering condition.
Thus, this set $P_j^t$ is listed in the state $(g, P^t_{k+1},\dots,P^t_m)$, and the latter by definition belongs to $Q'_-$.
Then the transition from $p$ is fine,
and the transition from $(g, P^t_{k+1},\dots,P^t_m)$ is also fine since none of the $S_j^t$'s has a state from $Q'_+$.

Let us now analyze the complexity of the algorithm.
The original saturation algorithm by definition uses states in $Q'_1, \dots, Q'_{n+1}$.
We show that in fact a subset of those states is actually needed,
by constructing a sequence $Q''_1, \dots, Q''_{n+1}$ which is pointwise included in the former,
and such that $Q''_1 \cup \cdots \cup Q''_{n+1}$ is one exponential smaller than $Q'_1 \cup \cdots \cup Q'_{n+1}$.
We define the following sets:
\begin{align*}
	&Q''_{n+1}=Q'_{n+1},\qquad
	Q''_n=Q'_n,\\
	&Q''_{n-1}=Q''_n\cup\setof{(g,P_{k+1},\dots,P_m)\in Q'_{n-1}}
		{\forall i.|P_i|\geq 1,\ \sum_i|P_i\setminus\{p^n\}|\leq 1},\\
	&Q''_i = Q''_{i+1}\cup\bigcup_{g\in\Rules_i}\set g\times\left(2^{Q''_{i+1}}\right)^{m-k}\qquad\mbox{for $1\leq i\leq n-2$},
\end{align*}
where $g=(p,a(u_1,\dots,u_k,b(v_1,\dots,v_{m'}),u_{k+1},\dots,u_m)\rewritesto S,r)$ and $\ord(b) = i$.
Thus to $Q''_{n-1}$ we only add those states $(g,P_{k+1},\dots,P_m)$ where all $P_i$ except one are equal to $\{p^n\}$, and the remaining one is either a singleton or a pair $\{q,p^n\}$ for some state $q$.
We can see that all transitions in $\Delta'$ will only use states from $Q_1''$.
To prove this, recall that $\Delta'$ is the least set containing $\Delta\cup\Delta_0$ and closed under applying the saturation procedure.
The initial set $\Delta\cup\Delta_0$ only uses states from $Q'_{n+1}\subseteq Q_1''$.
Thus, suppose that the current set of transitions uses only states from $Q_1''$;
we should prove that transitions added by (one step of) the saturation procedure also use only states from $Q_1''$.
Notice that all states in the considered run tree $t$ on $r$ appear in some transition, so they come from $Q_1''$. 
The only ``new'' state is $(g,P_{k+1},\dots,P_m)$ created for a deep rule.
Let $b$ be the root of the special subtree of the left side of $g$.
If $\ord(b)\neq n-1$, then we have $(g,P_{k+1},\dots,P_m)\in Q''_{\ord(b)}\subseteq Q''_1$,
since, as already observed, $P_{k+1}, \dots, P_m \subseteq Q''_1\cap Q'_{\ord(b)+1}=Q''_{\ord(b)+1}$.
For $\ord(b)=n-1$, we recall that each of the sets $P_{k+1}, \dots, P_m$ describes a subtree of order $n$, 
so it is either of the form $\{p^n\}$ or $\bigcup t(r^{-1}(\var x))$ for some variable $\var x$ of order $n$.
But, as observed previously, at most one of these sets may contain a state from $Q'_+$.
However (as required in an non-$n$-alternating automaton) all states of order $n$ except $p^n$ are from $Q'_+$.
Thus $\sum_{j=k+1}^{m}|P_j\setminus\{p^n\}|\leq 1$,
and in consequence $(g,P_{k+1},\dots,P_m)\in Q_{n-1}''\subseteq Q_1''$.
This finishes the proof that only states from $Q_1''$ are used.

Notice that $|Q''_{n-1}|\leq |Q'_n|+|\calR|\cdot2(m-1)\cdot|Q_n'|$ (while $|Q'_{n-1}|$ is exponential in $|Q_n'|$),
where $m$ is the maximal rank of any symbol in $\Sigma$.
Consequently,
$|Q''_1|\leq\mathsf{exp}_{n-2}(O((|Q|+|\Rules|)\cdot m\cdot|\Rules|))$ and $|\Delta'|\leq\mathsf{exp}_{n-1}(O((|Q|+|\Rules|)\cdot m\cdot|\Rules|))$.

\section{The translation for annotated pushdown systems}

We present graphically the rewrite rules of the resulting ordered tree transition system.
The rules in Figure~\ref{fig:apds:otpds:rules} are the same as in the main text.
We hope that the graphical presentation better conveys the intuition behind them.

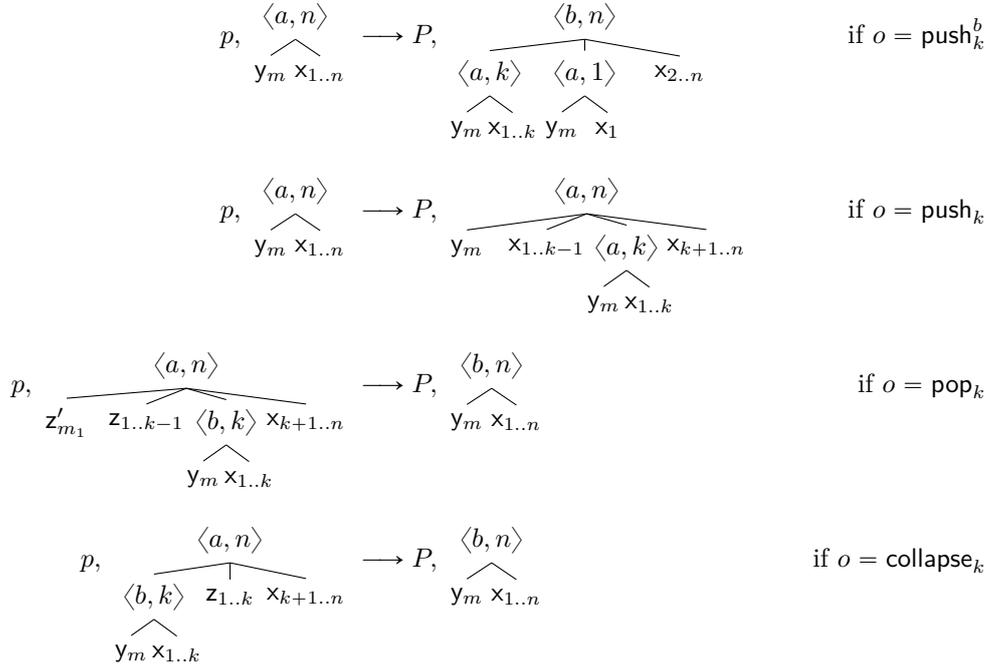
\begin{figure}
	\begin{align*}
		p, \begin{tikzpicture}[baseline=-2.3ex,scale=0.5,level/.style={sibling distance = 1.3cm}]
			\node {$\tuple{a, n}$}
			 	child{ node {${\var y_m}$}}
				child{ node {$\var x_{1..n}$}};
		\end{tikzpicture}
		&\longgoesto
		P, \begin{tikzpicture}[baseline=-2.3ex,scale=0.5,level 1/.style={sibling distance = 2.5cm},level 2/.style={sibling distance = 1.2cm}]
			\node {$\tuple{b, n}$}
				child{ node {$\tuple{a, k}$}
					child{ node {${\var y_m}$}}
					child{ node {$\var x_{1..k}$}}
				}
				child{ node {$\tuple{a, 1}$}
					child { node {${\var y_m}$}}
					child { node {$\var x_1$}}
				}
				child{ node {$\var x_{2..n}$}};
		\end{tikzpicture}
		&\textrm{ if } o = \apush k b
		\\[1ex]
		p, \begin{tikzpicture}[baseline=-2.3ex,scale=0.5,level/.style={sibling distance = 1.3cm}]
			\node {$\tuple{a, n}$}
				child{ node {${\var y_m}$}}
				child{ node {$\var x_{1..n}$}};
		\end{tikzpicture}
		&\longgoesto
		P, \begin{tikzpicture}[baseline=-2.3ex,scale=0.5,level 1/.style={sibling distance = 2.1cm},level 2/.style={sibling distance = 1.2cm}]
			\node {$\tuple{a, n}$}
				child{ node {${\var y_m}$}}
				child{ node {$\var x_{1..k-1}$}}
				child{ node {$\tuple{a, k}$}
					child{ node {${\var y_m}$}}
					child{ node {$\var x_{1..k}$}}
				}
				child{ node {$\var x_{k+1..n}$}};
		\end{tikzpicture}
		&\textrm{ if } o = \apush k {}
		\\[1ex]
		p, \begin{tikzpicture}[baseline=-2.3ex,scale=0.5,level 1/.style={sibling distance = 2.1cm},level 2/.style={sibling distance = 1.2cm}]
			\node {$\tuple{a, n}$}
			    	child{ node {${\var z'_{m_1}}$}}
				child{ node {$\var  z_{1..k-1}$}}
				child{ node {$\tuple{b, k}$}
					child { node {${\var y_m}$}}
					child { node {$\var x_{1..k}$}}
					}
			    	child{ node {$\var x_{k+1..n}$}};
		\end{tikzpicture}
		&\longgoesto
		P, \begin{tikzpicture}[baseline=-2.3ex,scale=0.5,level/.style={sibling distance = 1.3cm}]
			\node {$\tuple{b, n}$}
				child{ node {${\var y_m}$}}
				child{ node {$\var x_{1..n}$}};
		\end{tikzpicture}
		&\textrm{ if } o = \apop k
		\\[1ex]
		p, \begin{tikzpicture}[baseline=-2.3ex,scale=0.5,level 1/.style={sibling distance = 2cm},level 2/.style={sibling distance = 1.2cm}]
			\node {$\tuple{a, n}$}
				child{ node {$\tuple{b, k}$}
					child { node {${\var y_m}$}}
					child { node {$\var x_{1..k}$}}
				}
				child{ node {$\var z_{1..k}$}}
			    	child{ node {$\var x_{k+1..n}$}};
		\end{tikzpicture}
		&\longgoesto
		P, \begin{tikzpicture}[baseline=-2.3ex,scale=0.5,level/.style={sibling distance = 1.3cm}]
			\node {$\tuple{b, n}$}
				child{ node {${\var y_m}$}}
				child{ node {$\var x_{1..n}$}};
		\end{tikzpicture}
		&\textrm{ if } o = \acollapse k
	\end{align*}
	
	\caption{Translation from annotated higher-order pushdown systems to AOTPSs}
	\label{fig:apds:otpds:rules}
	
\end{figure}
\vfill


\section{The translation for Krivine machines}

\label{app:krivine:safety}

We present graphically the rewrite rules of the resulting ordered
tree transition system. The rules in
Figure~\ref{fig:km:otpds:rules} are the same as in the main text. We
hope that the graphical presentation better conveys the
intuition behind them.

\begin{figure}
	
	\begin{align}
		\tag*{(var)}
		\label{rule:var}
		&p, \begin{tikzpicture}[baseline=-2.3ex,scale=0.5,level/.style={sibling distance = 1.2cm}]
			\node {$\marked {x_i^\alpha}$}
				child { node {$\var z_1$} }
				child { node {$\cdots$} edge from parent[draw=none] }
				child { node {$\var z_{i-1}$} }
			    child { node {$M^\alpha$}
					child { node {$\vec {\var x}$} } }
				child { node {$\var z_{i+1}$} }
				child { node {$\cdots$} edge from parent[draw=none] }
				child { node {$\var z_n$} }
				child { node {$\vec {\var y}$} };
		\end{tikzpicture}
		\longgoesto
		\set p, \begin{tikzpicture}[baseline=-2.3ex,scale=0.5,level/.style={sibling distance = 1.2cm}]
			\node {$\marked {M^\alpha}$}
			  	child { node {$\vec {\var x}$} }
				child { node {$\vec {\var y}$} };
		\end{tikzpicture}
		\\
		\tag*{(app)}
		\label{rule:app}
		&p, \begin{tikzpicture}[baseline=-2.3ex,scale=0.5,level/.style={sibling distance = 1.2cm}]
			\node {$\marked {M^\alpha N^{\alpha_1}}$}
			  	child { node {$\vec {\var x}$} }
				child { node {$\var y_2$} }
			    child { node {$\cdots$} edge from parent[draw=none] }
			    child { node {$\var y_k$} };
		\end{tikzpicture}
		\longgoesto
		\set p, \begin{tikzpicture}[baseline=-2.3ex,scale=0.5,level/.style={sibling distance = 1.5cm}]
			\node {$\marked {M^\alpha}$}
			  	child { node [left=1ex]{\makebox[1ex][l]{$\restrict {\vec {\var x}} {M^\alpha}$} } }
				child { node {$N^{\alpha_1}$}
					child { node {\parbox[b][1ex]{1ex}{$ \restrict {\vec {\var x}} {N^{\alpha_1}} $}} } }
				child { node {$\var y_2$} }
			    child { node {$\cdots$} edge from parent[draw=none] }
			    child { node {$\var y_k$} };
		\end{tikzpicture}
		\\
		\tag*{(fix)}
		\label{rule:fix}
		&p, \begin{tikzpicture}[baseline=-2.3ex,scale=0.5,level/.style={sibling distance = 1.2cm}]
			\node {$\marked {Y M^{\alpha \arrow \alpha}}$}
			  	child { node {$\vec {\var x}$} }
				child { node {$\vec {\var y}$} };
		\end{tikzpicture}
		\longgoesto
		\set p, \begin{tikzpicture}[baseline=-2.3ex,scale=0.5,level/.style={sibling distance = 1.5cm}]
			\node {$\marked {M^{\alpha \arrow \alpha}}$}
			  	child { node {$\vec {\var x}$} }
				child { node {$Y M$}
					child { node {$\vec {\var x}$} } }
				child { node {$\vec {\var y}$} };
		\end{tikzpicture}
		\\
		\tag*{(abs$\mbox{}$)}
		&p, \begin{tikzpicture}[baseline=-2.3ex,scale=0.5,level/.style={sibling distance = 1.2cm}]
			\node {$\marked {\lambda x_i^{\alpha_0}. M^\alpha}$}
			  	child { node {$\vec {\var x}$} }
				child { node {$\var y_0$} }
				child { node {$\vec {\var y}$} };
		\end{tikzpicture}
		\longgoesto
		\set p, \begin{tikzpicture}[baseline=-2.3ex,scale=0.5,level/.style={sibling distance = 1.3cm}]
			\node {$\marked {M^\alpha}$}
			  	child { node {$\var x_1$} }
				child { node {$\cdots$} edge from parent[draw=none] }
			  	child { node {$\var x_{i-1}$} }
				child { node {$\var y_0$} }
			  	child { node {$\var x_{i+1}$} }
				child { node {$\cdots$} edge from parent[draw=none] }
				child { node {$\var x_n$} }
				child { node {$\vec {\var y}$} };
		\end{tikzpicture}
		\ignore{
			\tag*{(abs$\mbox{}_2$)}
			\label{rule:abs2}
			p, \begin{tikzpicture}[baseline=-2.3ex,scale=0.5,level/.style={sibling distance = 1.2cm}]
				\node {$\marked {\lambda x^{\alpha_0} \cdot M^\alpha}$}
				  	child { node {$\vec {\var x}$} }
					child { node {$\var y_0$} }
					child { node {$\vec {\var y}$} };
			\end{tikzpicture}
			&\longgoesto
			\set p, \begin{tikzpicture}[baseline=-2.3ex,scale=0.5,level/.style={sibling distance = 1.2cm}]
				\node {$\marked {M^\alpha}$}
				  	child { node {$\vec {\var x}$} }
					child { node {$\vec {\var y}$} };
			\end{tikzpicture}
			\textrm{ if } x^{\alpha_0} \not\in \freevars{M^\alpha}
			\\
		}
		\\
		\tag*{(const$\mbox{}_1$)}
		&p, \begin{tikzpicture}[baseline=-2.3ex,scale=0.5,level/.style={sibling distance = 1.5cm}]
			\node {$\marked {a^{0^k \arrow 0}}$}
				child { node {$\vec {\var x}$} }
				child { node {$\vec {\var y}$} };
		\end{tikzpicture}
		\longgoesto
		\set{(1, P_1), \dots, (k, P_k)}, \begin{tikzpicture}[baseline=-2.3ex,scale=0.5,level/.style={sibling distance = 1.5cm/#1}]
			\node {$\marked {a^{0^k \arrow 0}}$}
				child { node {$\vec {\var x}$} }
				child { node {$\vec {\var y}$} };
		\end{tikzpicture}
		\quad\qquad \forall (p \trans a P_1 \cdots P_k) \in \Delta
		\\
		\tag*{(const$\mbox{}_2$)}
		&(i, P_i), \begin{tikzpicture}[baseline=-2.3ex,scale=0.6,level/.style={sibling distance = 1.2cm}]
			\node {$\marked {a^{0^k \arrow 0}}$}
				child { node {$\vec {\var z}$} }
				child { node {$\var y_1$} }
			  	child { node {$\cdots$} edge from parent[draw=none] }
				child { node {$\var y_{i-1}$} }
				child { node {$M_i^0$} edge from parent
					child { node {$\vec {\var x}$} } }
				child { node {$\var y_{i+1}$} }
			  	child { node {$\cdots$} edge from parent[draw=none] }
				child { node {$\var y_k$} };
		\end{tikzpicture}
		\longgoesto
		P_i, \begin{tikzpicture}[baseline=-2.3ex,scale=0.5,level/.style={sibling distance = 1.5cm/#1}]
			\node {$M_i^0$}
				child { node {$\vec {\var x}$} };
		\end{tikzpicture}
	\end{align}

	\caption{Translation from the Krivine machine to AOTPSs}
	\label{fig:km:otpds:rules}
	
\end{figure}

\ignore{

\paragraph{Safety.}

We show a connection between a special class of Krivine machine and safe AOPDS (as defined in Sec.~\ref{sec:otpds}).
A term $M$ is \emph{locally unsafe} if it has a free variable of level strictly smaller than $\level M$,
otherwise it is \emph{locally safe}.
An occurrence of a term $N$ as a subterm of a term $M$ is \emph{safe} if it is in the context $\dots (N\,K)\dots$.
A term is \emph{safe},if all occurrences of locally unsafe subterms are safe.
An alternating Krivine Machine with states $\KM = \tuple {l, \Gamma, Q, M^0, \Delta}$ is \emph{safe} if $M^0$ is safe.
A configuration $(p,C^\tau,C_1^{\tau_1},\dots,C_k^{\tau_k})$ is \emph{safe} if in all its closures (also recursively inside environments), except potentially the head closure, skeletons are safe.
Notice that when the machine is safe, the skeleton of the head closure has to be a subterm of a safe term, but it need not to be safe itself (all occurrences of locally unsafe proper subterms are safe, but the whole term may be locally unsafe).
Observe that every transition of a safe Krivine machine from a safe configuration necessarily lead to a safe configuration.
Indeed, in the rule for application $M\,N$ we create a new argument closure with term $N$, but the appearance of $N$ in $M\,N$ is not safe, so $N$ is locally safe hence safe.
Similarly, in the rule for $Y\,M$ we create a new argument closure with term $YM$, but the appearance of $M$ in $Y\,M$ is not safe, so $M$ is locally safe, and then $Y\,M$ is also locally safe since it has the same free variables and smaller level.
All other rules change only the head closure.
%
It is easy to see that a safe configuration of a Krivine machine is encoded in a safe tree of the AOPDS.
Indeed, $\enc {} {N^\tau,\rho}$ starts with a node of order $\ord(\tau)$,
and, since $N$ is safe (locally safe),
in $\rho$ we have only closures of levels not smaller than $\level\tau$;
the encodings of these closures will have orders not greater than $\ord(\tau)$.
The head closure (whose skeleton need not to be locally safe) is not encoded in this way,
but the root will be labelled by $[N]$ letter, which is of maximal order $l$.
Additionally, the rules of the created AOTPS are safe as well.

}

\section{Encoding AOTPS into the Krivine machine with states}

\label{app:aotps2km}

\thmaotpstokm*

The encoding is performed in four steps.

\subsection*{Step 1: Flattening the AOTPS}
\label{app:flat}

\newcommand{\br}[1]{\llbracket #1\rrbracket}
\newcommand{\sub}{\mathsf{sub}}

In this step we show that the l.h.s.{\@} of rewrite rules of AOTPSs can be flattened,
in the sense that the only lookahead that the system has is in deep rules,
and all other subtrees are just variables.
%
%
We recall that a rule $p,l\to S,r$ is \emph{flat} if $l$ is either of the form $a(\var x_1,\dots,\var x_m)$ (shallow rule) or $a(\var x_1,\dots,\var x_{k-1},b(\var y_1,\dots,\var y_{m'}),\var x_{k+1},\dots,\var x_m)$ (deep rule);
an AOTPS is \emph{flat} when all its rules are flat.

\begin{theorem}
	\label{thm:flat}
	Every AOTPS $\Sys$ can be converted into an equivalent flat AOTPS $\Sys'$ of exponential size.
\end{theorem}

Fix an AOTPS $\calS=\langle n,\Sigma,P,\calR\rangle$.
We create an equivalent flat AOTPS $\calS'=\langle n,\Gamma,P,\calR'\rangle$ as follows.
Let $\sub(\calR)$ contain all proper (that is, other than the whole tree) subtrees of $l$ for all rules $p,l\to S,r$ in $\calR$.
The intuition is that the new system will store in each node with $k$ children a tuple consisting of $k$ subsets of $\sub(\calR)$; the $i$-th of them will contain these patterns that match to the tree at the $i$-th child.
In this way, when a rule $p,l\to S,r$ is applied,
it is sufficient to read in the root whether appropriate subpatterns of $l$ match to subtrees starting in children of the root, instead of testing a longer part of the tree.
Thus, as the new alphabet $\Gamma$ we take $\bigcup_{a\in\Sigma}\{a\}\times(2^{\sub(\calR)})^{\rank(a)}$,
where the rank and the order of a symbol in $\Gamma$ is inherited from its $\Sigma$ coordinate.
For a $\Sigma$-tree $t$ we obtain the $\Gamma$-tree $\enc {} t$ by labelling each node $u$ with the tuple of $\rank(t(u))$ sets where the $i$-th set contains those trees $l'\in\sub(\calR)$ 
that match the subtree rooted at the $i$-th child of $u$ (that is, for which this subtree equals $l'\sigma$ for some substitution $\sigma$).

Consider some rule $p,l\to S,r$.
Thanks to the additional labeling of $\enc {} t$,
seeing only the label in the root of $t$ we know whether $l$ matches to $t$.
This allows us to replace every $r$-ground subtree of $l$ by an $r$-ground variable, and obtain $l$ in one of the two forms allowed in a flat rule.
On the other hand, we have to ensure that $r$ creates a tree with the correct labelling.
This is possible since our labelling is compositional: 
the tuple in a node can be computed basing on labels of its children.

More precisely, for each rule $p,l\to S,r$ we create new rules as follows.
Let $l=a(u_1,\dots,u_m)$, and, if the rule is deep, let $u_k=b(v_1,\dots,v_{m'})$ be the lookahead subtree of $l$.
For all sets $X_1,\dots,X_m\subseteq\sub(\calR)$ such that $u_1\in X_1,\dots,u_m\in X_m$, and (in the case of a deep rule) for all sets $Y_1,\dots,Y_{m'}\subseteq\sub(\calR)$ such that $v_1\in Y_1,\dots,v_{m'}\in Y_{m'}$ there will be one new rule.\footnote{%
	Some sets in $\sub(\calR)$ are ``inconsistent'', as well as some sets $Y_1,\dots,Y_{m'}$ may be ``inconsistent'' with $X_k$.
	New rules using such sets are redundant, but it also does not hurt to add them, as anyway they cannot be applied to any tree of the form $\enc {} t$.}
If $p,l\to S,r$ is shallow, as the new left side we take $(a,X_1,\dots,X_m)(\var x_1,\dots,\var x_m)$, where we leave $\var x_i=u_i$ if $u_i$ is a variable, and we take a fresh variable $\var x_i\not\in\calV(r)$ of order $\ord(u_i)$ otherwise (when $u_i$ is $r$-ground).
If $p,l\to S,r$ is deep, as the new left side we take $(a,X_1,\dots,X_m)(\var x_1,\dots,\var x_{k-1},(b,Y_1,\dots,Y_{m'})(\var y_1,\dots,\var y_{m'}),\var x_{k+1},\dots,\var x_m)$,
where again we leave subtrees being variables and we replace other subtrees by fresh variables.
The choice of $X_1,\dots,X_m$ and $Y_1,\dots,Y_{m'}$ assigns sets of ``matching patterns'' to all variables:
\begin{align*}
	&\mathsf{mp}(\var x_i)=X_i,&&\mathsf{mp}(\var y_i)=Y_i.
\end{align*}
Then in a bottom-up manner we can assign such sets to all subtrees of $r$:
\begin{align*}
	&\mathsf{mp}(c(w_1,\dots,w_j))=\{\var z\in\sub(\calR)\mid\ord(\var z)=\ord(c)\}\cup\\
	&\qquad\cup\{c(w_1',\dots,w_j')\in\sub(\calR)\mid w_1'\in\mathsf{mp}(w_1),\dots,w_j'\in\mathsf{mp}(w_k)\}.
\end{align*}
The new right side of the rule is $\br{r}$ with $\br\cdot$ defined by:
\begin{align*}
	&\br{w}=(c,\mathsf{mp}(w_1),\dots,\mathsf{mp}(w_j))(\br{w_1},\dots,\br{w_j}) &&\mbox{if }w=c(w_1,\dots,w_j),\\
	&\br{w}=w &&\mbox{if $w$ is a variable}.
\end{align*}
%
%
We remark that in order to recover correct marking of the right side of a rule it was necessary to mark a node by patterns matching to children of that node instead of patterns matching to the node itself (the latter marking would be insufficient).
It is easy to see that each transition of $\langle\calC_\calS,\to_\calS\rangle$
can be faithfully simulated by a transition of $\langle\calC_{\calS'},\to_{\calS'}\rangle$.

\ignore{ 

	It remains to argue that in the non-deterministic case,
	we can avoid the exponential blow-up in the construction of the alphabet $\Sigma\times 2^{\sub(\calR)}$.
	Indeed, when alternation is present,
	in general we need to decorate the root of a tree with all the elements in $\calR$
	that could be possibly matched in later application of rules.
	In the presence of alternation, there can be several such later applications, and all different.
	However, when no alternation is present, i.e., the system is non-deterministic,
	it suffices to guess in each node the only single decoration that will be later used in the rule testing and destroying this node (notice that when a node is matched by some non-variable node of $l$, then this node will not be present in the next configuration).
	Thus, for non-deterministic systems it suffices to take as the new alphabet $\Sigma \times \sub(\calR)$.
	The construction is modified accordingly in order to guess a single decoration from $\sub(\calR)$
	instead of carrying all of them in parallel.
	Therefore, we can transform a given non-deterministic AOTPS into a (non-deterministic) flat one with only a polynomial blow-up.
	Here the equivalence is of a weaker form than in the general (alternating) case: a single configuration of the original system corresponds to several configurations of the new system (with different guesses); 
	nevertheless such equivalence is sufficient for our proof of Theorem \ref{thm:complexity:nondeterministic} in the next section.
	
}

\subsection*{Step 2: Eliminating control locations}

To ease the presentation in step 3, we now remove control locations from the AOTPS.
To allow alternation, we have to extend slightly the definition of an AOTPS.
The rules will be now of the form $l\rewritesto R$, where $R$ is a set of trees $r$ such that $l\rewritesto r$ is a rewrite rule.
The resulting alternating transition system $\tuple {\mathcal C_\Sys, \rewritesto_\Sys}$ has $\Sigma$-trees as configurations, 
and, for every configuration $t$, and set of configurations $U$ there is a transition $t\rewritesto_\Sys U$ if 
there exists a rule $l\rewritesto R$ of $\Sys$ and a substitution $\sigma$ s.t.{\@} $t = l\sigma$ and $U = \setof{r\sigma}{r\in R}$. 

Control locations can be encoded in the root symbol of the pushdown tree.
The new alphabet is $\Sigma' = \Sigma \cup (\Sigma \times P)$,
where new symbols in $\Sigma \times P$ inherit order and rank from the $\Sigma$-component.
Let $p, l \rewritesto S, r$ be a shallow rule of the original system,
with $l = a(\var x_1, \dots, \var x_m)$ and $r = c(s_1, \dots, s_k)$ (the case for a deep rule is analogous).
Then, the new system has a rule $l' \rewritesto R$ (with no control locations),
with $l' = (a, p)(u_1, \dots, u_m)$ and $R = \setof{(c, q)(s_1, \dots, s_k)}{q\in S}$.
There is a problem when $r$ is just a variable $\var x$ (necessarily occurring in $l$).
In this case, we use a deep rule by guessing the root symbol $c$ at the (unique) position of $\var x$ in $l$ (below $a$).
That is, for every shallow rule $p, l \rewritesto S, \var x_k$ of the original system,
and for every symbol $c$ of rank $h$,
we introduce a deep rule $l' \rewritesto R$,
where $l' = (a, p)(\var x_1, \dots, \var x_{k-1}, c(\var y_1, \dots, \var y_h), \var x_{k+1}, \dots, \var x_m)$
and $R = \setof{(c, q)(\var y_1, \dots, \var y_h)}{q\in S}$.
We can assume that in the original system there are no deep rules of the form $p, l \rewritesto S, \var x$,
as those can be broken down into two rules, a deep one where the r.h.s.{\@} is not a variable,
and a shallow one where the r.h.s.{\@} is a variable.
Thus, we do not have consider any other case when removing control locations.
%

Notice that while starting from a flat AOTPS, the obtained AOTPS without control locations is also flat.

\subsection*{Step 3: From AOTPSs to higher-order recursion schemes}

%
After the first two steps we have a flat AOTPS without control locations (where rules are of the form $l\rewritesto R$ with $R$ a set of trees). 
Our goal is to translate an AOTPS of this form into a Krivine machine,
thus proving Theorem~\ref{thm:aotps2km}.
In order to obtain a more natural translation, we use recursion schemes, a model quite similar to the Krivine machine.

We first define alternating higher-order recursion schemes with states.
We use types as defined in Sec.~\ref{sec:krivine:machine}.
However, instead of $\lambda Y$-terms, we use applicative terms, where moreover we can use typed nonterminals from some set $\calN$.
An \emph{applicative term} is either (i) a constant $a^{0^{\rank(a)}\to 0}\in\Gamma$, (ii) a variable $x^\alpha\in\calV$, (iii) a nonterminal $A^\alpha\in\calN$, (iv) an application $(M^{\alpha\to\beta}\,N^\alpha)^\beta$.
An \emph{alternating recursion scheme with states} of level $n\in\Nat_{>0}$ is a tuple $\calG=\langle n,\Gamma,Q,\calN,\calR,\Delta\rangle$, where $\langle\Gamma,Q,\Delta\rangle$ is an alternating tree automaton,
$\calN$ is a finite set of nonterminals of level at most $n$,
and $\calR$ is a function assigning to each nonterminal $A$ in $\calN$ of type $\alpha_1\to\dots\to\alpha_k\to 0$ a rule of the form
$A\,x_1^{\alpha_1}\,\dots\,x_k^{\alpha_k}\to M$, where $M$ is an applicative term of type $0$ with free variables in $\{x_1^{\alpha_1},\dots,x_k^{\alpha_k}\}$, constants from $\Gamma$, and nonterminals from $\calN$.
A recursion scheme $\calG$ induces an alternating transition system $\langle\calC_\calG,\to_\calG\rangle$, 
where in a configuration $(p,M)\in\calC_\calG$ we have $p\in Q$ and $M$ is a closed applicative term of type $0$ using constants from $\Gamma$ and nonterminals from $\calN$.
We have two kinds of transitions.
First, for each rule $A\,x_1^{\alpha_1}\,\dots\,x_k^{\alpha_k}\to M$ in $\calR$ we have a transition 
$(p,A\,M_1\,\dots\,M_k)\to_\calG \{(p,M[M_1/x_1,\dots,M_k/x_k])\}$.
Second, for every $(p\xrightarrow{a}P_1\dots P_k)\in\Delta$ we have a transition $(p,a\,M_1\,\dots\,M_k)\to_\calG (P_1\times\{M_1\})\cup\dots\cup (P_k\times\{M_k\})$.

Fix a flat AOTPS $\calS=\langle n,\Sigma,\calR\rangle$ without control locations.
An \emph{extended letter} is a pair $(a,o)$ where $a\in\Sigma$ and $o\colon\{1,\dots,\rank(a)\}\to\{1,\dots,n\}$.
The meaning is that the letter is $a$ and its children have orders $o(1),\dots,o(\rank(a))$.
%
For each extended letter $(a,o)$ we will have a corresponding nonterminal $A_{a,o}$.
A first approximation of the encoding is that a tree $a(u_1,\dots,u_k)$ will be represented as $A_{a,o}$ to which encodings of $u_1,\dots,u_k$ are applied as arguments.
Then the rule for the nonterminal $A_{a,o}$ can simulate all shallow rules of our system $\calS$, constructing any term having $u_1,\dots,u_k$ as subterms.
Notice that rules of $\calS$ are flat, so we need not to look inside $u_1,\dots,u_k$; however we need to know their orders---that is why we assign nonterminals to extended letters not just to letters.

Nevertheless, there are also deep rules.
In order to handle them, each tree $a(u_1,\dots,u_k)$ will be represented in multiple ways.
One representation will handle shallow rules.
Moreover, for each extended letter $(b,o')$ and $\ihat$ we will have a nonterminal $A_{a,o}^{b,o',\ihat}$, 
and $a(u_1,\dots,u_k)$ will be represented as $A_{a,o}^{b,o',\ihat}$ with applied representations of $u_1,\dots,u_k$.
This nonterminal will be still waiting for subtrees of a potential parent of our $a$, having label $b$; when they will be applied, the nonterminal will simulate deep rules of $\calS$ having $b(\dots,a(\dots),\dots)$ on the left side,
where the $a$ is on the $\ihat$-th position.
As we have multiple encodings of a tree, we need to keep them in parallel.
Thus $A_{a,o}$ (and similarly $A_{a,o}^{b,o',\ihat}$) instead of taking one argument for each subtree,
takes multiple arguments for each subtree, one for each encoding of the subtree.

Let $\overline\Sigma_k$ be the set of all triples $(b,o',\ihat)$, where $(b,o')$ is an extended letter for $b\in\Sigma$, and $o'(\ihat)=k$.
Let us fix some (arbitrary) total order on elements of $\overline\Sigma_k$, that will be used to order arguments of our terms.
Using the product notation $\alpha\times\beta\to\gamma$ for the type $\alpha\to\beta\to\gamma$, we define types 
\begin{align*}
	\Psi_k=0\times\prod_{(a,o,\ihat)\in\overline\Sigma_k}\alpha_{a,o,>k},
\end{align*}
where the $(a,o,\ihat)$ are ordered according to our fixed order on $\overline\Sigma_k$.
We have used here the types
\begin{align*}
	\alpha_{a,o,>k}=\Psi_{o(i_1)}\to\dots\to\Psi_{o(i_m)}\to 0,
\end{align*}
where $i_1<\dots<i_m$ is the list of all $i\in\{1,\dots,\rank(a)\}$ for which $o(i)>k$.
In particular we have $\alpha_{a,o,>n}=0$.
Notice that $\alpha_{a,o,>k}$ and types in $\Psi_k$ are of level at most $n-k$.
We are now ready to define the type of our nonterminals: $A_{a,o}$ has type
\begin{align*}
	\Psi_{o(1)}\to\dots\to\Psi_{o(\rank(a))}\to 0
\end{align*}
and $A_{a,o}^{b,o',\ihat}$ has type
\begin{align*}
	\Psi_{o(1)}\to\dots\to\Psi_{o(\rank(a))}\to \alpha_{b,o',>\ord(a)}.
\end{align*}

Recall that the resulting higher-order recursion scheme has to use constants and an alternating tree automaton to simulate alternation.
We will have two kinds of constants in $\Gamma$: $\vee_i$ of type $0^i\to 0$, and $\wedge_i$ of type $0^i\to 0$, defined for appropriate numbers $i\in\Nat$ (we need only finitely many of them).
The constant $\vee_i$ simulates nondeterministic choice, and $\wedge_i$ simulates universal choice.
The tree automaton $\langle\Gamma,Q,\Delta\rangle$ is defined in the natural way: $Q$ consists of a single state $q$, 
and we have transitions
\begin{align*}
	q\xrightarrow{\vee_i}\emptyset\dots\emptyset Q\emptyset\dots\emptyset,\qquad q\xrightarrow{\wedge_i}Q\dots Q.
\end{align*}
In particular there is no transition reading $\vee_0$, and we have the transition $q\xrightarrow{\wedge_0}\varepsilon$.

Next, we define our encoding of trees into applicative terms.
In the definition below we use tuples just as a shorthand: $M\,(N_1,\dots,N_k)$ is intended to mean $M\,N_1\,\dots\,N_k$.
A tree is encoded as a tuple of applicative terms:
\begin{align*}
	\Enc(u)=(\encP(u),\encP^{b_1,o_1,\ihat_1}(u),\dots,\encP^{b_k,o_k,\ihat_k}(u))
\end{align*}
where $(b_1,o_1,\ihat_1)<\dots<(b_k,o_k,\ihat_k)$ are all the elements of $\overline\Sigma_{\ord(u)}$
ordered according to the fixed order on $\overline\Sigma_{\ord(u)}$.
We have here one basic encoding:
\begin{align*}
	\encP(a(u_1,\dots,u_m))=A_{a,o}\,\Enc(u_1)\,\dots\,\Enc(u_k)\qquad\mbox{where }\forall i \cdot o(i)=\ord(u_i).
\end{align*}
Additionally for each $(b_j,o_j,\ihat_j)\in\overline\Sigma_{\ord(a)}$ we have 
\begin{align*}
	&\encP^{b_j,o_j,\ihat_j}(a(u_1,\dots,u_m))=A^{b_j,o_j,\ihat_j}_{a,o}\,\Enc(u_1)\,\dots\,\Enc(u_k)
	\qquad\mbox{where } \forall i \cdot o(i)=\ord(u_i).
\end{align*}

It remains to define rules for the nonterminals from the rules of $\calS$.
Let us use the convention on variable naming that every left side of a rule in $\calR$ is of the form $a(\var x_1,\dots,\var x_m)$ or $b(\var y_1,\dots,\var y_{\ihat-1},a(\var x_1,\dots,\var x_m),\var y_{\ihat+1},\dots,\var x_k)$.
Recall that $\calS$ is flat, so every left side is of such form. 
A right side $R=\{r_1,\dots,r_l\}$ of a rule is encoded as
\begin{align*}
	\encP(R)=\wedge_l\,\encP(r_1)\,\dots\,\encP(r_l),
\end{align*}
where $\encP(r_i)$ is defined as for normal trees, with the encoding of variables given by:
\begin{align*}
	&\encP(\var x_i)=x_i,&&\encP^{a,o,\ihat}(\var x_i)=x_i^{a,o,\ihat},\\
	&\encP(\var y_i)=y_i,&&\encP^{a,o,\ihat}(\var y_i)=y_i^{a,o,\ihat}.
\end{align*}
Here, $\var x_i$ is a variable of the AOTPS, while $x_i$ is a variable in an applicative term of the recursion scheme,
and similarly for the other variables.
In order to define a rule for a nonterminal $A_{a,o}$ we look at all rules with $a(\var x_1,\dots,\var x_m)$ on the left side with $\ord(\var x_i)=o(i)$ for all $i$.
Let $R_1,\dots,R_l$ be the right sides of these rules.
Then we take 
\begin{align*}
	A_{a,o}\,\Enc(\var x_1)\,\dots\,\Enc(\var x_m)\to\vee_{l+m}\,\encP(R_1)\,\dots\,\encP(R_l)\,\deep_1\,\dots\,\deep_m
\end{align*}
with 
\begin{align*}
	&\deep_{\ihat}=x_{\ihat}^{a,o,\ihat}\,\Enc(\var x_{i_1})\,\dots\,\Enc(\var x_{i_d}),
\end{align*}
where $i_1<\dots<i_d$ are all the $i\in\{1,\dots,m\}$ for which $o(i)>o(\ihat)$.
In order to define a rule for a nonterminal $A^{b,o',\ihat}_{a,o}$ we look at all rules having on the left side $b(\var y_1,\dots,\var y_{\ihat-1},a(\var x_1,\dots,\var x_m),\var y_{\ihat+1},\dots,\var x_k)$
with $\ord(\var y_i)=o'(i)$ and $\ord(\var x_i)=o(i)$ for all $i$, and $\ord(a)=o'(\ihat)$ 
(in particular for nonterminals $A^{b,o',\ihat}_{a,o}$ with $\ord(a)\neq o'(\ihat)$ there are no such rules).
Let $R_1,\dots,R_l$ be the right sides of these rules (when $\ord(a)\neq o'(\ihat)$, we always have $l=0$).
Then we take 
\begin{align*}
	A_{a,o}^{b,o',\ihat}\,\Enc(\var x_1)\,\dots\,\Enc(\var x_m)\,\Enc(\var y_{i_1})\,\dots\,\Enc(\var y_{i_d})\to\vee_l\,\encP(R_1)\,\dots\,\encP(R_l),
\end{align*}
where $i_1<\dots<i_d$ are all the $i\in\{1,\dots,k\}$ for which $o'(i)>\ord(a)$.

Let us briefly see the correspondence between steps of $\calS$ and steps of $\calG$.
Consider a tree $u=a(u_1,\dots,u_m)$.
Its encoding $\encP(u)$ is $A_{a,o}\,\Enc(u_1)\,\dots\,\Enc(u_m)$ where $o(i)=\ord(u_i)$ for each $i$.
The first step of $\calG$ performed from $\encP(u)$ results in the right side of the rule for $A_{a,o}$ where we substitute $\encP(u_i)$ for $x_i$ and $\encP^{b,o',\ihat}(u_i)$ for $x_i^{b,o',\ihat}$ for each $i,b,o',\ihat$
(denote this substitution as $[\varphi]$).
The resulting tree starts with some $\vee_j$, so then $\calG$ nondeterministically chooses one of its subtrees.
There is a subtree $\encP(R)[\varphi]$ for each shallow rule $a(\var x_1,\dots,\var x_m)\to R$ of $\calS$ where $\ord(\var x_i)=o(i)=\ord(u_i)$ for each $i$.
These are exactly all shallow rules that can be applied to $u$.
By applying this rule to $u$ we obtain the set $\{r\sigma\mid r\in R\}$ where $\sigma$ maps $\var x_i$ to $u_i$ for each $i$.
Notice that $\encP(R)[\varphi]$ starts with $\wedge_{|R|}$ and has as children $\encP(r)[\varphi]$ for each $r\in R$.
Thus $\calG$ in the next step will transit into the set $\{\encP(r)[\varphi]\mid r\in R\}$.
Finally, we see that $\encP(r\sigma)=\encP(r)[\varphi]$.

Another possibility for $\calG$ in the second step is to transit for some $\ihat\in\{1,\dots,m\}$ to 
\begin{align*}
	\deep_{\ihat}[\varphi] &= (x_{\ihat}^{a,o,\ihat}\,\Enc(\var x_{i_1})\,\dots\,\Enc(\var x_{i_d}))[\varphi] = \encP^{a,o,\ihat}(u_{\ihat})\,\Enc(u_{i_1})\,\dots\,\Enc(u_{i_d})=\\
	&=A^{a,o,\ihat}_{b,o'}\,\,\Enc(v_1)\,\dots\,\Enc(v_k)\,\Enc(u_{i_1})\,\dots\,\Enc(u_{i_d}),
\end{align*}
where $u_{\ihat}=b(v_1,\dots,v_k)$, and $o'(i)=\ord(v_i)$ for each $i$, and $i_1<\dots<i_d$ are all $i\in\{1,\dots,m\}$ for which $o(i)>o(\ihat)$. 
The next three steps of $\calG$ simulate exactly all deep rules of $\calS$ having on the left side $a(\var y_1,\dots,\var y_{\ihat-1},b(\var x_1,\dots,\var x_k),\var y_{\ihat+1},\dots,\var y_m)$
with $\ord(\var y_i)=o(i)$ and $\ord(\var x_i)=o'(i)$ for all $i$
(in the same way as it was for shallow rules and for the nonterminal $A_{a,o}$).
It is important that the right sides of such rules use only those variables $\var y_i$ for which $\ord(\var y_i)=o(i)>o(\ihat)=\ord(b)$.
%
%
We conclude that $\langle\calC_\calS,\to_\calS\rangle$ and $\langle\calC_{\calG},\to_{\calG}\rangle$ faithfully simulate each other.

\subsection*{Step 4: From higher-order recursion schemes to the Krivine machine}

It is well-known how to translate from one formalism to the other; cf.~\cite{Salvati:Walukiewicz:MSCS:2015}.

\ignore{

	\paragraph{Safety.}

	Finally, let we observe that the translation from AOTPSs to higher-order recursion schemes above also preserves safety, after small adjustment.
	For applicative terms we define safety in the same way as for $\lambda Y$-terms.
	A rule $A\,x_1\,\dots\,x_k\to M$ is safe if $M$ is \emph{safe},
	and a higher-order recursion scheme with states is \emph{safe} if all its rules are safe.

	We have to adjust the translation so that a tree of order $k$ will be encoded in a term of level exactly $n-k$ (now it is encoded in a term of level \emph{at most} $n-k$).
	For that we need spare nonterminals of each level.
	We define a type $\beta_k$ of level $k$:
	\begin{align*}
		\beta_0=0,&&\beta_k=\beta_{k-1}\to0\qquad\mbox{for $k>0$.}
	\end{align*}
	We also have a nonterminal $B_k$ of type $\beta_k$ with rules $B_0\to B_0$ and $B_k\,x\to B_k\,x$ for $k>0$.
	We modify types in the encoding as follows:
	\begin{align*}
		&\Psi_k=\beta_{n-k}\times\prod_{(a,o,\ihat)\in\overline\Sigma_k}\alpha_{a,o,>k},\\
		&\alpha_{a,o,>k}=\Psi_{o(i_1)}\to\dots\to\Psi_{o(i_m)}\to \beta_{n-k}.
	\end{align*}
	Now we have that $\alpha_{a,o,>k}$ and all types in $\Psi_k$ are of level exactly $n-k$.
	Nonterminals $A_{a,o}$'s have type
	\begin{align*}
		\Psi_{o(1)}\to\dots\to\Psi_{o(\rank(a))}\to \beta_{n-\ord(a)}.
	\end{align*}
	The definition of the encoding of a tree is left unchanged.
	We need to modify the rules for nonterminals.
	First, for those nonterminals that should now take more arguments (all nonterminals encoding letters of order smaller than $n$) we add an additional spare argument $z$ that will never be used on the right side.
	Second, arguments to $\vee_i$ and $\wedge_i$ should have type $0$, but now an encoding of a tree has some type $\beta_k$.
	We correct it by passing the nonterminal $B_{k-1}$ as an argument.
	Thus in the definition of $\encP(R)$ we use $(\encP(r_i)\,B_{n-\ord(r_i)-1})$ instead of $\encP(r_i)$ if $\ord(r_i)<n$, and we leave $\encP(r_i)$ if $\ord(r_i)=n$.
	Also in the definition of $\deep_{\ihat}$ we append $B_{n-o(\ihat)-1}$ whenever $o(\ihat)<n$.
	The dynamics of the system is as previously, we just sometimes pass some $B_k$ as an additional argument that is ignored.

	Let us see that a $(\Sigma\cup\calV)$-tree $t$ will be encoded into safe terms, by induction on the structure.
	Of course encodings of variables are safe.
	Let $t=a(u_1,\dots,u_k)$.
	Terms in $\Enc(u_i)$ are safe by induction assumption.
	Safety of $t$ implies that all variables used in $t$ are of order at most $\ord(t)$;
	thus all variables used in $\Enc(u_i)$ are of order at least $n-\ord(t)$.
	Encodings of $t$ (that is $\encP(t)$ and all $\encP^{b,o',\ihat}(t)$) are of level $n-\ord(t)$, so they are locally safe.
	Finally, beyond subterms of $\Enc(u_i)$, the only subterms of $\encP(t)$ that do not occur safely are the whole term and $A_{a,o}$, so $\encP(t)$ is safe; similarly $\encP^{b,o',\ihat}(t)$.
	We conclude that the right side of every nonterminal is safe, since every its subterm that does not occur safely is either of level $0$, or is a single nonterminal or a variable or a constant, or is a subterm of the encoding of some $r_i$.
	Thus the obtained higher-order recursion scheme with states is safe.

}

\section{Ordered annotated multi-pushdown systems}

\label{sec:oapds:app}

We encode ordered annotated multi-pushdown systems \cite{Hague:FSTTCS:2013} into AOTPSs.
Formally, an \emph{alternating ordered annotated multi-pushdown system} is a tuple $\OAMPDS = \tuple {m, n, \Gamma, Q, \Delta}$,
where $m \in \Natpos$ is the number of higher-order pushdowns,
$n \in \Natpos$ is the order of each of the $m$ higher-order pushdowns,
$\Gamma$ is a finite pushdown alphabet containing a distinguished initial symbol $e$,
$Q$ is a finite set of control locations,
and $\Delta \subseteq Q \times\{1,\dots,m\}\times\Gamma\times O_n \times 2^Q$
is a set of rules, where 
%
$O_n =  \bigcup_{k = 1}^n \setof {\apush {k} b, \apush {k} {}, \apop {k}, \acollapse {k}} {b \in \Gamma}$.
Intuitively, a rule $(p,l,a,o,P)$ can be applied when the control location is $p$ and the topmost symbol on the $l$-th stack is $a$, and it applies the stack operation specified by $o$ to this stack.
Pop and collapse operations are called \emph{consuming}.
An alternating ordered annotated multi-pushdown system $\OAMPDS$
induces an alternating transition system $\tuple {\mathcal C_\OAMPDS, \goesto_\OAMPDS}$, where
$\mathcal C_\OAMPDS = Q \times \Gamma_n^m$,
and $(p, w_1, \dots, w_m) \goesto_\OAMPDS P \times \set {(w_1', \dots, w_m')}$
if, and only if, there exists a rule $(p, l, a, o, P) \in \Delta$ s.t.
\begin{enumerate}[1)]
	\item $w_l = \stack {a^{u_l}, \cdots}$,
	\item if $o$ is consuming, then $w_1' = \dots = w_{l-1}' = \stack{e^{\stack{\,}},\stack{\,},\dots,\stack{\,}}$,
	\item if $o$ is not consuming, then $w_1' = w_1, \dots, w_{l-1}' = w_{l-1}$,
	\item $w_l' = o(w_l)$, and
	\item $w_{l+1}' = w_{l+1}, \dots, w_m' = w_m$.
\end{enumerate}
For $c\in\mathcal C_\OAMPDS$ and $T \subseteq Q$,
the \emph{(control-state) reachability problem} for $\OAMPDS$ asks whether
$c\in \prestar {T \times \Gamma_n^m}$.

\paragraph{Encoding.}

Let $\Sigma$ be an ordered alphabet containing, 
for every pushdown index $l\in\{1,\dots,m\}$ and order $k\in\{1,\dots,n\}$,
\begin{inparaenum}[1)]
\item an end-of-stack symbol $(l, k, \bot)$ of order $(l-1) \cdot n + k$ and rank $0$,
\item a symbol $(l, k)$ of order $(l-1) \cdot n + k$ and rank $k+2$,
\item a symbol $(l, \bullet)$ of order $(l-1) \cdot n + 1$ and rank $n+2$.
\end{inparaenum}
Moreover, $\Sigma$ contains a symbol $\bullet$ of order $1$ and rank $m \cdot (n+1)$, and, 
for every pushdown index $l\in\{1,\dots,m\}$ and every $a \in \Gamma$, a symbol $(l,a)$ of order $(l-1)\cdot n+1$ and rank $0$.
Thus $\Sigma$ has order $mn$.
Notice that the size of $\Sigma$ is $\size \Sigma = O(m\cdot (n+\size \Gamma))$.
Fix a pushdown index $l$.
An empty order-$k$ pushdown is encoded as the tree $\enc{l,k}{\stack{\,}}=(l,k,\bot)$.
A nonempty order-$k$ pushdown $\stack{a^{\hat u}, u_1, \dots, u_k}$ is encoded as the tree 
\begin{align*}
	\enc {l, k} {\stack{a^{\hat u}, u_1, \dots, u_k}} &= 
	\begin{tikzpicture}[baseline=-2.3ex,scale=0.6,level/.style={sibling distance = 2.7cm}]
		\node {$(l, k)$}
			child{ node {$(l, a)$}}
			child{ node {$\encB l {\hat u}$}}
			child{ node {$\enc {l, 1} {u_1}$}}
		    child{ node {$\cdots$} edge from parent[draw=none]}
			child{ node {$\enc {l, k} {u_k}$}};
	\end{tikzpicture},
\end{align*}
and
\begin{align*}
	\encB l {\stack{b^{\hat u}, u_1, \dots, u_n}} &= 
	\begin{tikzpicture}[baseline=-2.3ex,scale=0.6,level/.style={sibling distance = 2.7cm}]
		\node {$(l, \bullet)$}
			child{ node {$(l, b)$}}
			child{ node {$\encB l {\hat u}$}}
			child{ node {$\enc {l, 1} {u_1}$}}
		    child{ node {$\cdots$} edge from parent[draw=none]}
			child{ node {$\enc {l, n} {u_n}$}};
	\end{tikzpicture}.
\end{align*}
To simplify the encoding, we assume w.l.o.g.{\@} that collapse links are always of order $n$ (collapse to a lower order is still allowed though).
An $m$-tuple of nonempty order-$n$ pushdowns
\[ w = (\stack {a_1^{\hat u_1}, u_{1, 1}, \dots, u_{1, n}}, \dots, \stack {a_m^{\hat u_m}, u_{m,1},\dots, u_{m, n}}) \]
is encoded as the following tree $\enc{} {w}$
%
\ignore
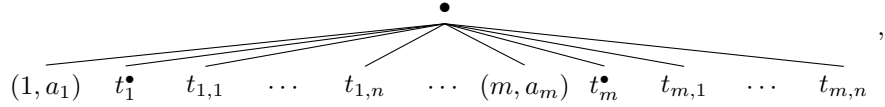
{
\begin{align*}
	\begin{tikzpicture}[baseline=-2.3ex,scale=0.7,level/.style={sibling distance = 2.2cm}]
		\node {$\bullet$}
			child{ node {$(1, a_1)$}}
			child{ node {$\encB 1 {\hat u_1}$}}
			child{ node {$\enc {1, 1} {u_{1, 1}}$}}
		    child{ node {$\cdots$} edge from parent[draw=none]}
			child{ node {$\enc {1, n} {u_{1, n}}$}}
		    child{ node {$\cdots$} edge from parent[draw=none]}
			child{ node {$(m,a_m)$}}
			child{ node {$\encB m {\hat u_m}$}}
			child{ node {$\enc {m, 1} {u_{m, 1}}$}}
		    child{ node {$\cdots$} edge from parent[draw=none]}
			child{ node {$\enc {m, n} {u_{m, n}}$}};
	\end{tikzpicture}
\end{align*}
}
\begin{align*}
	\begin{tikzpicture}[baseline=-2.3ex,scale=0.7,level/.style={sibling distance = 1.5cm}]
		\node {$\bullet$}
			child{ node {$(1, a_1)$}}
			child{ node {$t^\bullet_1$}}
			child{ node {$t_{1,1}$}}
		    child{ node {$\cdots$} edge from parent[draw=none]}
			child{ node {$t_{1,n}$}}
		    child{ node {$\cdots$} edge from parent[draw=none]}
			child{ node {$(m, a_m)$}}
			child{ node {$t^\bullet_m$}}
			child{ node {$t_{m,1}$}}
		    child{ node {$\cdots$} edge from parent[draw=none]}
			child{ node {$t_{m,n}$}};
	\end{tikzpicture},
\end{align*}
where $t^\bullet_l = \encB l {\hat u_l}$ for every $l = 1, \dots, m$,
and $t_{l,k} = \enc {l, k} {u_{l, k}}$ for every $l = 1, \dots, m$ and $k = 1, \dots, n$.
%
We can observe here two differences when comparing to the encoding of annotated higher-order pushdown systems.
First, we do not put the topmost stack symbol in the root, but in an additional child of the form $(l, a)$ just below the root. 
This allows to avoid an alphabet of exponential size containing all $m$-tuples of letters from $\Gamma$
(the rest of the stack is encoded in an analogous way, but this is only for uniformity).
Second, at the beginning of each collapse link we use a special letter $(l,\bullet)$ instead of $(l,k)$.
This letter has a fixed order, and thanks to that we can use a variable of this fixed order to match the subtree encoding a collapse link.
Recall that for annotated higher-order pushdown systems we have created separate rules for each possible order of the collapse links,
but here such a solution would result in an exponential blowup since we have $m$ collapse links of independent orders.

Let $\tuple {m, n, \Gamma, Q, \Delta}$ be an ordered annotated multi-pushdown system.
We define an equivalent AOTPS $\Sys = \tuple {mn, \Sigma, Q, \Rules_\Sys}$
of order $mn$,
where $\Sigma$ and $Q$ are as defined above,
and $\Rules_\Sys$ contains a rule for each rule in $\Delta$.
%
%
We use the convention that variables $\var x_{l, k}$, $\var z_{l, k}$ have order $(l-1) \cdot n + k$, and 
variables $\var y_l$, $\var y'_l$, $\var t_l$, are of order $(l-1) \cdot n + 1$.
We write $\var x_l^{i..j}$ (with $i \leq j$) for the tuple of variables $(\var x_{l, i}, \dots, \var x_{l, j})$.
If $(p, l, a, \apush k b, P) \in \Delta$, then there is the following shallow rule in $\Rules_\Sys$:
\begin{align*}
	&p, \begin{tikzpicture}[baseline=-.0ex]
		\Tree [. $\bullet$
			[. $\var t_1$ ]
			[. $\var y_1$ ]
			[. $\var x_1^{1..n}$ ]
			[. $\cdots$ ]
			[. $(l,a)$ ]
			[. $\var y_l$ ]
			[. $\var x_l^{1..n}$ ]
			[. $\cdots$ ]
			[. $\var t_m$ ]
			[. $\var y_m$ ]
			[. $\var x_m^{1..n}$ ]
		]
	\end{tikzpicture}
	\\ &\longgoesto
	P, \!\!\!\!\!\!\!\!\!\!\!\! \begin{tikzpicture}[baseline=-.0ex]
		\Tree [. $\bullet$
			[. $\var t_1$ ]
			[. $\var y_1$ ]
			[. $\var x_1^{1..n}$ ]
			[. $\cdots$ ]
			[. {$(l, b)$} ]
			[. {$(l, \bullet)$}
				[. {$(l, a)$} ]
				[. $\var y_l$ ]
				[. $\var x_l^{1..n}$ ]
			]
			[. {$(l, 1)$}
				[. {$(l, a)$} ]
				[. $\var y_l$ ]
				[. $\var x_{l, 1}$ ]
			]
			[. $\var x_l^{2..n}$ ]
			[. $\cdots$ ]
			[. $\var t_m$ ]
			[. $\var y_m$ ]
			[. $\var x_m^{1..n}$ ]
		]
	\end{tikzpicture}.
\end{align*}
If $(p, l, a, \apush k {}, P) \in \Delta$, then there is the following shallow rule in $\Rules_\Sys$:
\begin{align*}
	&p, \begin{tikzpicture}[baseline=-2.3ex,scale=0.6,level/.style={sibling distance = 1.3cm}]
		\node {$\bullet$}
					child{ node {$\var t_1$}}
				    child{ node {$\var y_1$}}
					child{ node {$\var x_1^{1..n}$}}
				    child{ node {$\cdots$} edge from parent[draw=none]}
					child{ node {$(l,a)$}}
	    			child{ node {$\var y_l$}}
					child{ node {$\var x_l^{1..n}$}}
				    child{ node {$\cdots$} edge from parent[draw=none]}
					child{ node {$\var t_m$}}
				    child{ node {$\var y_m$}}
				    child{ node {$\var x_m^{1..n}$}};
	\end{tikzpicture}
	\\ &\longgoesto
	P, \begin{tikzpicture}[baseline=-2.3ex,scale=0.6,level 1/.style={sibling distance = 1.6cm},level 2/.style={sibling distance = 1.4cm}]
		\node {$\bullet$}
			child{ node {$\var t_1$}}
		    child{ node {$\var y_1$}}
			child{ node {$\var x_1^{1..n}$}}
	    	child{ node {$\cdots$} edge from parent[draw=none]}
	    	child{ node {$(l,a)$}}
	    	child{ node {$\var y_l$}}
		    child{ node {$\var x_l^{1..k-1}$}}
			child{ node {$(l, k)$}
				child{ node {$(l,a)$}}
			    child{ node {${\var y}_l$}}
				child{ node {$\var x_l^{1..k}$}}
			}
		    child{ node {$\var x_l^{k+1..n}$}}
		    child{ node {$\cdots$} edge from parent[draw=none]}
			child{ node {$\var t_m$}}
		    child{ node {${\var y}_m$}}
		    child{ node {$\var x_m^{1..n}$}};
	\end{tikzpicture}.
\end{align*}
%
%
If $(p, l, a, \apop k, P) \in \Delta$, then there is the following deep rule in $\Rules_\Sys$:
\begin{align*}
	&p, \begin{tikzpicture}[baseline=-2.3ex,scale=0.6,level 1/.style={sibling distance = 1.6cm},level 2/.style={sibling distance = 1.3cm}]
		\node {$\bullet$}
			child{ node {$\var t_1$}}
		    child{ node {${\var y}_1$}}
			child{ node {$\var x_1^{1..n}$}}
	    	child{ node {$\cdots$} edge from parent[draw=none]}
			child{ node {$(l,a)$}}
		    child{ node {$\var y_l'$}}
			child{ node {$\var z_l^{1..k-1}$}}
			child{ node {$(l, k)$}
				child { node {$(l,b)$}}
				child { node {${\var y}_l$}}
				child { node {$\var x_l^{1..k}$}}
				}
			child{ node [left=-0.8cm] {$\var x_l^{k+1..n}$}}
		    child{ node [left=-0.5] {$\cdots$} edge from parent[draw=none]}
			child{ node {$\var t_m$}}
		    child{ node {${\var y}_m$}}
			child{ node {$\var x_m^{1..n}$}};
	\end{tikzpicture}
	\\& \longgoesto P, \begin{tikzpicture}[baseline=-2.3ex,scale=0.6,level/.style={sibling distance = 1.4cm}]
		\node {$\bullet$}
			child { node [left=5ex] {$(1,e)$}}
			child { node [left=1ex] {$(1,1,\bot)$}}
			child { node [left=2.4ex] {$\cdots$} edge from parent[draw=none]}
			child { node [left=-3ex] {$(l\!-\!1,n,\bot)$}}
			child { node {$(l,b)$}}
		    child { node {${\var y}_l$}}
			child { node {$\var x_l^{1..n}$}}
		    child { node {$\cdots$} edge from parent[draw=none]}
			child{ node {$\var t_m$}}
		    child{ node {${\var y}_m$}}
			child{ node {$\var x_m^{1..n}$}};
	\end{tikzpicture}
\end{align*}
This deep rule satisfies the ordering condition since $(l, k)$ has order $(l-1)\cdot n + k$,
and all variables $\var x_l^{k+1..n}$ that appear also on the r.h.s.{\@} have order strictly higher than $(l, k)$.
Notice that $y_l'$ has relatively low order, but it does not appear on the r.h.s.
%
%
Finally, if $(p, l, a, \acollapse k, P) \in \Delta$, then there is the following deep rule in $\Rules_\Sys$:
\begin{align*}
	&p, \begin{tikzpicture}[baseline=-2.3ex,scale=0.6,level 1/.style={sibling distance = 1.6cm},level 2/.style={sibling distance = 1.2cm}]
		\node {$\bullet$}
			child{ node {$\var t_1$}}
		    child{ node {${\var y}_1$}}
			child{ node {$\var x_1^{1..n}$}}
		    child{ node {$\cdots$} edge from parent[draw=none]}
			child{ node {$(l,a)$}}
			child{ node {$(l, \bullet)$}
				child { node {$(l,b)$}}
				child { node {$\varord {y} l {}$}}
				child { node {$\var x_l^{1..k}$}}
			}
			child { node {$\var z_l^{1..k}$}}
			child { node {$\var x_l^{k+1..n}$}}
			child{ node {$\cdots$} edge from parent[draw=none]}
			child{ node {$\var t_m$}}
		    child{ node {${\var y}_m$}}
			child{ node {$\var x_m^{1..n}$}};
	\end{tikzpicture}
	\\ &\longgoesto P, \begin{tikzpicture}[baseline=-2.3ex,scale=0.6,level/.style={sibling distance = 1.5cm}]
		\node {$\bullet$}
			child{ node [left=6ex] {$(l,e)$}}
			child { node [left=2ex] {$(1,1,\bot)$}}			
			child { node [left=3ex] {$\cdots$} edge from parent[draw=none]}
			child { node [left=-3ex] {$(l\!-\!1,n,\bot)$}}
			child{ node {$(l,b)$}}
			child { node {$\varord {y} l {}$}}
			child { node {$\var x_l^{1..n}$}}
			child{ node {$\cdots$} edge from parent[draw=none]}
			child{ node {$\var t_m$}}
		    child{ node {${\var y}_m$}}
			child{ node {$\var x_m^{1..n}$}};
	\end{tikzpicture}.
\end{align*}
The deep rule above satisfies the ordering condition since $(l, \bullet)$ has order $(l-1)\cdot n + 1$,
and all variables $\var x_l^{k+1..n}$, $\var t_{l+1}$, $\var y_{l+1}$, $\var x_{l+1}^{1..n}$, $\dots$, $\var t_m$, $\var y_m$, $\var x_m^{1..n}$ that appear also on the r.h.s.{\@}
have strictly higher order.

\begin{lemma}[Simulation]
	We have that $(p, w) \goesto_\OAMPDS^* P \times \set {w'}$ if, and only if,
	$(p, \enc{} w) \goesto_{\Sys}^* P \times \set {\enc{} {w'}}$.
	Thus, the (control-state) reachability problem for $\OAMPDS$ is equivalent to the reachability problem for $\Sys$.
\end{lemma}

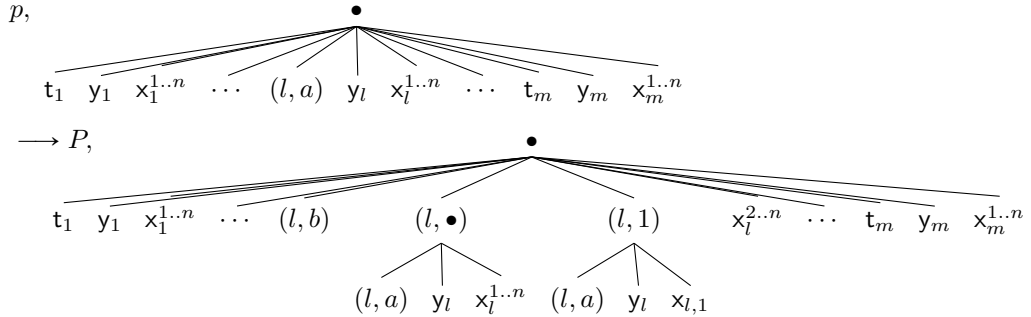
\end{document}